\documentclass[lettersize,journal]{IEEEtran}

\usepackage[T1]{fontenc}
\usepackage{aecompl} 
\usepackage{caption}

\ifCLASSINFOpdf
   \usepackage[pdftex]{graphicx}

   \graphicspath{{{\vartheta}_{k}^v/pdf/}{{\vartheta}_{k}^v/jpeg/}}

   \DeclareGraphicsExtensions{.pdf,.jpeg,.png}
\else

   \usepackage[dvips]{graphicx}

   \graphicspath{{{\vartheta}_{k}^v/eps/}}
  
   \DeclareGraphicsExtensions{.eps}
\fi

\usepackage{amsmath}
\usepackage{subeqnarray}
\usepackage{cases}
\usepackage{cite}
\usepackage{graphicx}
\usepackage{subfigure}
\usepackage{tabularx,booktabs}
\usepackage{dingbat}
\usepackage{diagbox}
\usepackage[ruled,linesnumbered]{algorithm2e}
\usepackage{amssymb}
\usepackage{algorithmic}
\usepackage{esint}
\usepackage{color}
\usepackage{lipsum}
\usepackage{cuted}
\usepackage{stfloats}
\usepackage{multirow}
\usepackage{ntheorem}
\usepackage{color}
\usepackage{threeparttable}
\usepackage{bbm}
\usepackage{tabu}
\usepackage{bbding}
\usepackage{tikz}
\usepackage{xcolor}
\usepackage{makecell}

\usepackage[colorlinks,
            linkcolor=black,       
            anchorcolor=black,  
            citecolor=black,        
            ]{hyperref}
\newcolumntype{C}{>{\centering\arraybackslash}X}
\theoremseparator{:}
\newtheorem{proposition}{\textcolor{black}{Proposition}}
\newtheorem{lemma}{Lemma}
\newtheorem{theorem}{Theorem}

\theoremstyle{nonumberplain}
\theorembodyfont{\normalfont}
\newtheorem{proof}{Proof}

\definecolor{newextractedpurple}{RGB}{127,0,127}

\usepackage[cmintegrals]{newtxmath}

\begin{document}

\title{Prioritized Information Bottleneck Theoretic Framework with Distributed Online Learning for Edge Video Analytics}

\author{Zhengru~Fang,~
Senkang~Hu,~\IEEEmembership{Student Member,~IEEE,}
Jingjing~Wang,~\IEEEmembership{Senior Member,~IEEE,}\\
Yiqin~Deng, 
Xianhao~Chen,~\IEEEmembership{Member,~IEEE}
and Yuguang Fang,~\IEEEmembership{Fellow,~IEEE, ACM}%
\IEEEcompsocitemizethanks{\IEEEcompsocthanksitem Z. Fang, S. Hu, Y. Deng and Y. Fang are with the Department of Computer Science, City University of Hong Kong, Hong Kong. E-mail: \{zhefang4-c, senkang.forest\}@my.cityu.edu.hk, \{yiqideng, my.fang\}@cityu.edu.hk.
\IEEEcompsocthanksitem J. Wang is with the School of Cyber Science and Technology, Beihang University, Beijing 100191, China, and also with Hangzhou Innovation Institute, Beihang University, Hangzhou 310051, China. Email: drwangjj@buaa.edu.cn.
\IEEEcompsocthanksitem X. Chen is with the Department of Electrical and Electronic Engineering, the University of Hong Kong, Hong Kong. E-mail: xchen@eee.hku.hk.}
\thanks{This work was supported in part by the Hong Kong SAR Government under the Global STEM Professorship and Research Talent Hub,  the Hong Kong Jockey Club under the Hong Kong JC STEM Lab of Smart City (Ref.: 2023-0108), and the Hong Kong Innovation and Technology Commission under InnoHK Project CIMDA. The work of J. Wang was supported in part by the National Natural Science Foundation of China under Grant No. 62222101 and No. U24A20213, in part by the Beijing Natural Science Foundation under Grant No. L232043 and No. L222039, and in part by the Fundamental Research Funds for the Central Universities. The work of Y. Deng was supported in part by the National Natural Science Foundation of China under Grant No. 62301300. The work of X. Chen was supported in part by the Research Grants Council of Hong Kong under Grant 27213824. A preliminary version has been accepted for IEEE Global Communications Conference (GLOBECOM 2024)\cite{fang2024pib}.}
}
\markboth{IEEE/ACM Transactions on Networking}%
{Shell \MakeLowercase{\textit{et al.}}: Bare Demo of IEEEtran.cls for IEEE Communications Society Journals}

\maketitle
\begin{abstract}
  Collaborative perception systems leverage multiple edge devices, such as surveillance cameras or autonomous cars, to enhance sensing quality and eliminate blind spots. Despite their advantages, challenges such as limited channel capacity and data redundancy impede their effectiveness. To address these issues, we introduce the Prioritized Information Bottleneck (PIB) framework for edge video analytics. This framework prioritizes the shared data based on the signal-to-noise ratio (SNR) and camera coverage of the region of interest (RoI), reducing spatial-temporal data redundancy to transmit only essential information. This strategy avoids the need for video reconstruction at edge servers and maintains low latency. It leverages a deterministic information bottleneck method to extract compact, relevant features, balancing informativeness and communication costs. For high-dimensional data, we apply variational approximations for practical optimization. To reduce communication costs in fluctuating connections, we propose a gate mechanism based on distributed online learning (DOL) to filter out less informative messages and efficiently select edge servers. Moreover, we establish the asymptotic optimality of DOL by proving the sublinearity of its regrets. To validate the effectiveness of the PIB framework, we conduct real-world experiments on three types of edge devices with varied computing capabilities. Compared to five coding methods for image and video compression, PIB improves mean object detection accuracy (MODA) by 17.8\% while reducing communication costs by 82.65\% under poor channel conditions.
\end{abstract}

\section{Introduction}
\subsection{Background}
\IEEEPARstart{V}{ideo} analytics is rapidly transforming various sectors such as urban planning, retail analysis, and autonomous navigation by converting visual data streams into useful insights\cite{padmanabhan2023gemel}. A large number of video cameras produce vast amounts of video data continuously and often require real-time video streams\cite{dai2022respire}. Numerous developing applications such as remote patient care\cite{wang2023contactless}, video games\cite{10175560}, UAV sensing\cite{10415249,tang2024decentralizedsemanticcommunicationcooperative} and virtual and augmented reality depend on the efficient analysis of video data with minimal delay\cite{10015055}.

The increasing number of smart devices requires a computational paradigm shift towards edge computing. This approach involves processing data closer to its source, resulting in several benefits compared to traditional cloud-based paradigms, particularly reduced latency and bandwidth costs. Even a delay as short as one second can have disastrous consequences. For example, interactive applications such as online gaming and video conferencing require latencies below 100 ms to ensure real-time feedback and a seamless user experience\cite{kamarainen2017measurement}. Similarly, VR/AR applications demand extremely low latencies, often less than 20 ms, to prevent motion sickness and maintain a high-quality immersive experience\cite{9944868}. Utilizing remote cloud services for data processing can result in significant latency increases, often exceeding 100 ms\cite{corneo2021much}. Moreover, the importance of privacy, particularly in regions with strict data protection laws such as the General Data Protection Regulation (GDPR), makes edge computing even more attractive\cite{marelli2018scrutinizing}. According to the Ponemon Institute, 60\% of the companies express apprehension towards cloud security and decide to manage their own data on-site to mitigate potential risks\cite{Ponemon2021}.

{\color{black}A key aspect of video analytics, particularly for Bird's Eye View (BEV) applications, is accurately capturing pedestrian occupancy across multiple camera views\cite{fang2024pacp,hu2024fullscenedomaingeneralizationmultiagent}. BEV representations rely on precise spatial context to depict ground-level scenes in order to minimize occlusions, blind spots, and viewpoint discrepancies. Pedestrian occupancy data enables reliable identification and localization of individuals within a shared view, refining collaborative perception and enhancing tasks such as detection and prediction in complex, dynamic environments. }However, the integration of edge devices into video analytics also brings in many significant challenges\cite{shao2023task}. The computational demands of deep neural network (DNN) models, such as GoogLeNet\cite{al2017deep}, which requires about 1.5 billion operations per image classification, place a substantial burden on the limited processing capacities of edge devices\cite{gao2024localization}. Additionally, the outputs from high-resolution cameras increase the communication load. For example, a 4K video stream requires up to 18 Gbps of bandwidth to transmit raw video data, potentially overwhelming wireless networks\cite{yaqoob2020survey}. Therefore, we need to explore efficient video coding for compressing streamed videos. As shown in \mbox{Fig. \ref{fig:traditional-model}}, traditional compression methods reconstruct streaming frames through efficient entropy coding and motion prediction. However, there is still a large amount of less informative data being processed, wasting communication bandwidth. For instance, if the tasks involve human recognition or positioning, the reconstructed background of each frame might not be useful for the application.

As shown in Fig. \ref{fig:IB-model}, the information bottleneck (IB) framework is a feasible choice for task-oriented video compression, enabling a trade-off between communication cost and prediction accuracy for specific tasks. However, the current communication strategies for integrating edge devices into video analytics are not effective enough. One major issue is how to handle the computational complexity and transmission of redundant data generated from the overlapping fields of view (FOVs) of multiple cameras\cite{cui2023stitched}. In scenarios with dense camera deployments, up to 60\% of data can be redundant due to overlapping FOVs, which unnecessarily overburdens the network\cite{jiang2020reinforcement}. In addition, these strategies often lack adaptability in transmitting tailored data features based on Region of Interest (RoI) and signal-to-noise ratio (SNR), resulting in poor video fusion or alignment. These limitations can negatively impact collaborative perception, even making it less effective than single-camera setups\cite{fang2024pacp}.

In this paper, we aim to refine multi-camera video analytics by developing a strategy to prioritize wireless video transmissions. Our proposed Prioritized Information Bottleneck (PIB) strategy attempts to effectively leverage SNR and RoI to selectively transmit data features, significantly reducing computational load and data transmissions. Our method can decrease data transmissions by up to 82.65\%, while simultaneously enhancing the mean object detection accuracy (MODA) compared to current state-of-the-art techniques. This approach not only compresses data but also intelligently selects data for processing to ensure only relevant information is transmitted, thus mitigating noise-induced inaccuracies in collaborative sensing scenarios. This innovation sets a new benchmark for efficient and accurate video analytics at the edge.

\subsection{State-of-the-Art}
\label{sec:Related Work}
This subsection reviews advancements in edge video analytics, with an emphasis on the designs for communication-computing latency reduction. We explore the information bottleneck method to enhance task-oriented performance by minimizing data redundancy. Additionally, we investigate online learning for dynamic ROI management and perceptual quality.

\subsubsection{Edge Video Analytics}
\begin{figure}[t]
  \centering
  \subfigure[Traditional compression method with redundant data.]{
    \includegraphics[width=0.47\textwidth]{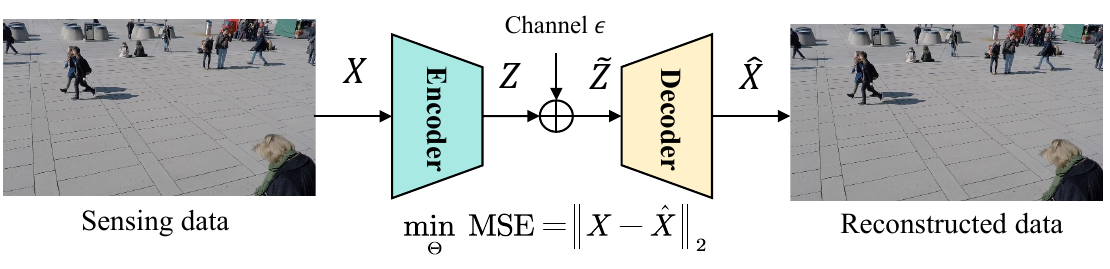}
    \label{fig:traditional-model}
  }
  \subfigure[Information bottleneck method for task-specific compression.]{
    \includegraphics[width=0.47\textwidth]{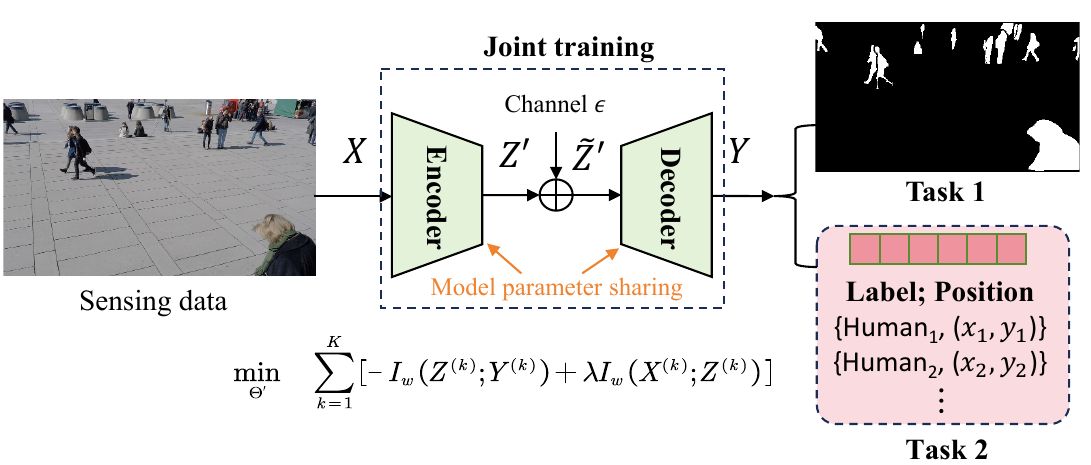}
    \label{fig:IB-model}
  }
  \caption{Comparison of compression methods: (a) Traditional compression method with redundant data, (b) Information bottleneck method for task-specific compression.}
  \label{fig:compression-methods}
\end{figure}

Live video analytics is crucial in various domains such as autonomous driving\cite{fang2024pacp,chen2023vehicle,hu2023adaptive, hu2024fullscenedomaingeneralizationmultiagent, hu2024agentscodriverlargelanguagemodel, hu2024collaborative}, mixed reality\cite{10217163,10090453}, 3D point cloud analytics\cite{10167694}, and traffic control\cite{10296872}. These applications, including object recognition tasks\cite{10101710}, are typically equipped with sophisticated machine learning models like Convolutional Neural Networks (CNNs) and Graph Neural Networks (GNNs)\cite{10018382,10526224}. However, offloading these applications to centralized clouds can result in unpredictable transmission delays in wide area networks, particularly when streaming high-quality videos\cite{10398474}. Therefore, researchers utilize edge computing to serve as a promising alternative to reduce service latency and energy consumption. Li \textit{et al.} propose a novel approach called \textit{ESMO} to optimize frame scheduling and model caching for edge video analytics\cite{10138921}. Khani \textit{et al.} introduce \textit{RECL}, a new framework for video analytics that integrates model reuse and online retraining to quickly adapt expert models to specific video scenes, optimizing resource allocation and achieving substantial performance gains over prior methods\cite{khani2023recl}. Wang \textit{et al.} design an MEC-enabled multi-device video analytics system using a Markov decision process to address real-time ground truth absence and content-varying degradation-accuracy issues, significantly enhancing the accuracy-latency tradeoff through adaptive information gathering and efficient bandwidth allocation. However, few works consider how to strike a dynamic balance between channel resources and inference performance in a multi-camera sensing system.

In typical edge video analytics scenarios, the lack of infrastructure and limited bandwidth makes real-time object detection challenging, especially for the multi-view camera sensing for wildlife monitoring or security surveillance in remote areas\cite{10025689}. To achieve real-time object sensing, it is crucial to reduce redundant information and the bandwidth demand. Semantic communication can address this challenge by transmitting only the essential semantic information, thereby compressing data streams and reducing transmission overhead. Zhang \textit{et al.} propose a comprehensive framework to highlight the importance of semantic communication in optimizing information transmission\cite{10634888}. Shao \textit{et al.} introduce a new conceptualization of semantic communication that characterizes it within joint source-channel coding theory, aiming to minimize the semantic distortion-cost region\cite{10540315}. Xie \textit{et al.} explore a deep learning-based semantic communication system with memory, showing how dynamic transmission techniques can enhance transmission reliability and efficiency by masking unessential elements\cite{10159023}. Zhou \textit{et al.} design and implement a deep learning-based image processing pipeline on the ESP32-CAM, proposing a DRL-based approach for efficient camera configuration adaptation in multi-camera systems\cite{9775607}. 
Existing research primarily focuses on rate-distortion (R-D) optimization, adapting the bitstream rate based on channel state information (CSI) to reconstruct raw videos\cite{10107793}. However, these methods rarely consider the performance of specific downstream tasks, such as mean object detection accuracy (MODA), as a system evaluation metric. Consequently, the transmitted information often contains redundant data.

To address this issue, researchers incorporate the information bottleneck (IB) framework to optimize edge video analytics by focusing on task-specific performance, thereby reducing redundancy\cite{10438074}. The IB framework helps in selectively transmitting only the most relevant features needed for specific tasks, enhancing efficiency. Pensia \textit{et al.} propose a novel feature extraction strategy in supervised learning that enhances classifier robustness to small input perturbations by incorporating a Fisher information penalty into the information bottleneck framework\cite{9088132}. Wang \textit{et al.} present a deep multi-view subspace clustering framework to extend the information bottleneck principle to a self-supervised setting, leading to superior performance in multi-view subspace clustering on real-world datasets\cite{10053658}. The IB tradeoff is well-suited for bandwidth-limited edge inference and is a key design principle in our study for efficient communication. Wang \textit{et al.} introduce the Informative Multi-Agent Communication (IMAC) method, which uses the information bottleneck principle to develop efficient communication protocols and scheduling for multi-agent reinforcement learning under limited bandwidth\cite{wang2020learning}. Shao \textit{et al.} propose a task-oriented communication scheme for multi-device cooperative edge inference, optimizing local feature extraction and distributed feature encoding to minimize data redundancy and focus on task-relevant information, leveraging the information bottleneck principle and extending it to a distributed deterministic information bottleneck framework. However, these existing studies often neglect the need to prioritize different data from multiple cameras for various downstream tasks, such as considering ROIs. Moreover, most existing studies overlook the correlation between multiple cameras in multi-view scenarios. Shao \textit{et al.} extract compact task-oriented representations based on the IB principle, but they overlook the fact that different tasks require varying levels of priority\cite{shao2023task}. By leveraging these correlations and levels of priority, it is possible to further reduce data rates by minimizing duplicate information across different camera feeds and enhance inference performance at the same time. 

\subsubsection{Learning-Based Transmission Scheduling}

Online learning in multi-agent deep reinforcement learning (MADRL) enhances multi-camera sensing under dynamic channels and overlapping ROIs. Effective transmission scheduling determines when and which agents communicate through binary vectors indicating allowed communications at specific time steps, forming a communication graph. Central transmission scheduling schemes use a globally shared policy to control communication. Kim \textit{et al.} propose SchedNet, which uses a global scheduler to limit broadcasting agents and reduce communication overhead\cite{kim2019schednet}. Du \textit{et al.} introduce FlowComm, forming a directed graph for communication\cite{du2021flowcomm}. Liu \textit{et al.} develop GA-Comm, using a two-stage attention network (G2ANet) to manage agent interactions\cite{liu2020gacomm}. Niu \textit{et al.} present MAGIC, a framework using a directed communication graph for enhanced coordination\cite{niu2021multi}. In distributed transmission scheduling schemes, each agent individually determines whether to communicate, forming a graph structure. Liu \textit{et al.} propose a framework for multi-agent collaborative perception, addressing communication group construction and decision-making for efficient bandwidth use, significantly reducing communication while maintaining performance\cite{liu2020when2com}. Deep learning optimizes these systems by refining communication actions and schedules, transmitting only relevant information, and minimizing redundancy. However, existing multi-camera cooperative sensing algorithms do not effectively address the transmission scheduling problem, particularly under dynamic wireless channels and overlapping ROIs.
\subsection{Our Contributions}

Edge computing plays a crucial role in collaborative perception systems, improving tracking precision and minimizing blind spots through multi-view sensing. However, challenges such as limited channel capacity and data redundancy impede their effectiveness. To address these issues, we propose the \underline{P}rioritized \underline{I}nformation \underline{B}ottleneck (\textbf{PIB}) framework for edge video analytics. Compared with the conference version\cite{fang2024pib}, this paper improves the MODA by up to 17.88\% and reduces the communication cost by 23.94\%. Our contributions are summarized as follows:

\begin{itemize}
  \item We propose the PIB framework that prioritizes the shared data based on the signal-to-noise ratio (SNR) and camera coverage of the region of interest (RoI), reducing redundancy both spatially and temporally. This approach avoids the need for video reconstruction at edge servers and maintains low latency.
  \item Our framework leverages a deterministic information bottleneck method to extract compact, relevant features, balancing informativeness and communication costs. For high-dimensional data, we apply variational approximations for practical optimization.
  \item To reduce communication costs in fluctuating links, we introduce a gate mechanism based on distributed online learning (DOL) to filter out unprofitable messages and efficiently select edge servers. We establish the asymptotic optimality of DOL by showing the sublinearity of its regrets.
  \item Our extensive experimental evaluations across different real-world hardware platforms demonstrate that PIB significantly enhances mean object detection accuracy (MODA) and reduces communication costs. Compared to TOCOM-TEM, JPEG, H.264, H.265, and AV1, PIB improves MODA by 17.8\% while reducing communication costs by 82.65\% under poor channel conditions. Additionally, our method can reduce the standard deviation of streaming packet sizes by up to 9.43\%, while simultaneously maintaining higher MODA, ensuring better transmission robustness under poor channel conditions.
\end{itemize}

The remainder of this paper is organized as follows: Sec. \ref{sec:System_Model} introduces the system model. Sec. \ref{sec:Problem Formulation} covers the problem formulation, including prioritized information bottleneck analysis and the CMAB problem. Sec. \ref{sec:Methodology} describes the methodology, focusing on the derivation of the IB problem's upper bound, loss function design, and the distributed gate mechanism. Sec. \ref{sec:Performance Evaluation} evaluates the performance of the PIB framework through simulations that forecast pedestrian occupancy in urban settings, considering communication bottlenecks, camera delays, and edge server connectivity.

\section{System Model}
\label{sec:System_Model}
\begin{figure}[t]
  \centering
  \includegraphics[width=0.48\textwidth]{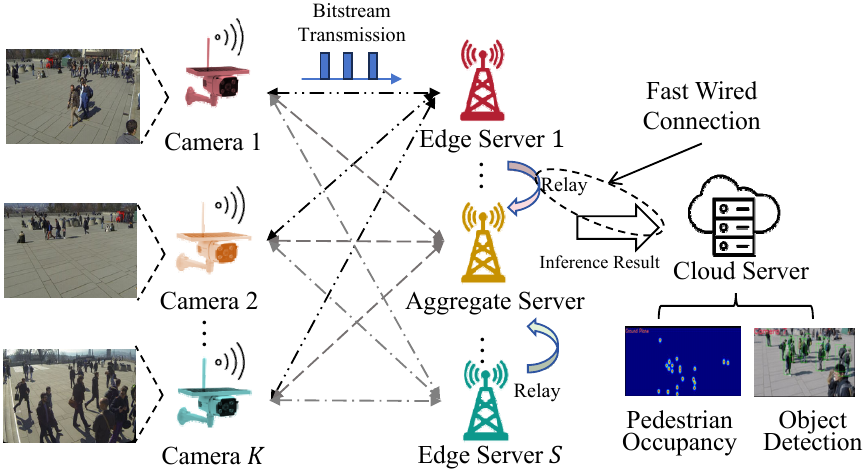}
  \caption{System model.}
  \label{fig:system model}
  \vspace{-3mm}
\end{figure}
{\color{black}{As illustrated in Fig. \ref{fig:system model}, our system comprises a set of edge cameras, denoted as \( \mathcal{K} = \{1, 2, \ldots, K\} \), and a set of edge servers, denoted as \( \mathcal{S} = \{1, 2, \ldots, S\} \). Each camera has a specific Field of view (FoV), \( \text{FoV}_k \), covering a subset of the monitored area. 
Our goal is to enable collaborative perception for pedestrian occupancy prediction under constrained channel capacity. The backbone of our system employs intermediate collaboration, where only encoded feature representations \(Z\) are transmitted to a central aggregate server for decoding and inference. Other edge servers serve solely as relay nodes, as they do not perform decoding or final inference due to the network structure which does not support task offloading for these operations. The aggregate server is chosen to minimize overall system delay, while the relay-only role of other edge servers ensures streamlined data flow. The key notations are given in Table \ref{tab:key_notations} for the ease of reading.}}
\begin{table}[ht]
\centering
\caption{{\color{black}{Key Notations}}}
\label{tab:key_notations}
\renewcommand{\arraystretch}{1.5} 
\color{black}
\begin{tabular}{|c|p{0.8\columnwidth}|} 
\hline
\multicolumn{1}{|c|}{\textbf{Notation}} & \multicolumn{1}{c|}{\textbf{Description}} \\ \hline
\( C_{k,s} \)                          & Capacity of the link between camera \( k \) and edge server \( s \)        \\ \hline
\( d_{k,s}^T \)                        & Transmission delay from camera \( k \) to edge server \( s \)             \\ \hline
\( d_{k,s_0}^I \)                      & Inference delay at the aggregate edge server \( s_0 \)                     \\ \hline
\( d_{r,s_0}^R \)                      & Relay delay from edge server \( r \) to aggregate server \( s_0 \)         \\ \hline
\( p_k \), \( w_k \)                   & Priority score and weight for camera \( k \)                               \\ \hline
\( X^{(k)} \)                          & Input data from camera \( k \)                                             \\ \hline
\( Z^{(k)} \)                          & Feature extracted from camera \( k \)                                      \\ \hline
\( Y \)                                & Random variable of the output                                              \\ \hline
\( \mathcal{E}_{k,s}^{c \rightarrow e} \) & Connection between camera \( k \) and edge server \( s \)                 \\ \hline
\( \mathcal{E}_{s}^{e \rightarrow e_0} \) & Connection between edge server \( s \) and aggregate server \( s_0 \)      \\ \hline
\( \mathcal{K} \)                      & Set of cameras selected at time \( t \)                                    \\ \hline
\( \Psi^F_s \)                         & Computational cost at edge server \( s \)                                  \\ \hline
\( \Psi^T_s \)                         & Computational cost of forwarding features                                  \\ \hline
\( \Psi^R_k \)                         & Remaining computing capacity of edge server \( k \)                        \\ \hline
\( \mathcal{T}_{k,s}^{c \rightarrow e} \) & Latency from camera \( k \) to edge server \( s \)                        \\ \hline
\( \mathcal{T}_{s}^{e \rightarrow e} \)   & Latency between edge server \( s \) and aggregate server \( s_0 \)         \\ \hline
\end{tabular}
\end{table}

\subsection{Communication Model}

We use Frequency Division Multiple Access (FDMA) to manage communication among cameras, defining the capacity \( C_{k,s} \) for each camera \( k \) and edge server \( s \) combination using the SNR-based Shannon capacity:
\begin{equation}
C_{k,s} = B_{k,s} \log_2 \left(1 + \text{SNR}_{k,s}\right),
\end{equation}
where \( B_{k,s} \) is the bandwidth allocated to the link between camera \( k \) and edge server \( s \), and \( \text{SNR}_{k,s} \) is the signal-to-noise ratio of this link. The transmission delay \( d_{k,s}^T \) is given by:
\begin{equation}
d_{k,s}^T = \frac{D}{C_{k,s}},
\end{equation}
where \( D \) is the data packet size. Each camera \( k \) decides whether to transmit data directly to its aggregate server or via a relay edge server based on channel quality: 1) If the channel quality is good, the camera transmits directly to the aggregate server $s_0$. The total delay is $d_{k,s}^{\text{total}} = d_{k,s_0}^T + d_{k,s_0}^I$, where \( d_{k,s_0}^I \) is the inference delay at the aggregate server. 2) If the channel quality is poor, the camera first transmits to a relay edge server \( r \), which then forwards the data to the aggregate server. The total delay in this case is $d_{k,s}^{\text{total}} = d_{k,r}^T + d_{r,s_0}^R + d_{k,s_0}^I$, where \( d_{r,s_0}^R \) is the relay delay. By dynamically choosing between relay and direct transmission, the system adapts to varying channel conditions, ensuring minimal delays and efficient use of network resources.

\subsection{Priority Weight Formulation}
\label{sec: Priority Weight Formulation}
Dynamic priority weighting is essential for optimizing network resource allocation, as various data sources require different levels of attention. Inspired by our previous work\cite{fang2024pacp}, we employ a dual-layer Multilayer Perceptron (MLP)\footnote{The MLP is trained in a supervised learning manner, where the input features are the normalized delay \({d}_{\text{norm}, k}\) and the normalized number of perceived moving objects \({\chi}_{\text{norm}, k}\). The target output is the optimal priority weight \(W_{\mathrm{target}}\). The loss function used for training is designed to minimize the discrepancy between the computed weights \(w_k\) and the target weights \(W_{\mathrm{target}}\), as described in Sec. \ref{sec:Network Loss Functions Derivation}.} to compute priority weights based on normalized delay and the number of perceived objects (\( \chi_k \)). 
\begin{equation}
p_k = \text{MLP}({d}_{\text{norm}, k}, {\chi}_{\text{norm}, k}; {\Theta_{M}}),
\end{equation}
where \( p_k \) denotes the computed priority score for camera \( k \), and \( {\Theta_{M}} \) represents the trainable parameters of MLP. The architecture of this MLP, featuring two layers, allows it to effectively model the interactions between delay and the number of perceived moving objects. Specifically, \( {d}_{\text{norm}, k} = \frac{{d}_k}{{d}_{\max}} \) and \( {\chi}_{\text{norm}, k} = \frac{{\chi}_k - {\chi}_L}{{{\chi}_U - {\chi}_L}} \). {\color{black}{To account for the dynamic nature of the system, we periodically update \( d_{\max} \). Specifically, we define \( d_{\max} \) as \(d_{\max} = \max_k (d_k) + \Delta\), where \( d_k \) is the delay for camera \( k \) and \( \Delta \) is a predefined threshold. This dynamic computation ensures that the normalized delay \( d_{\text{norm}, k} = \frac{d_k}{d_{\max}} \) remains within a reasonable range, preventing resource overcommitment. \( {\chi}_k \) represents the number of moving objects perceived by camera \( k \), while \( {\chi}_{U} \) and \( {\chi}_{L} \) denote the upper and lower bounds of the number of moving objects that any edge camera should perceive, respectively.}}

To transform the raw priority scores into a usable format within the system, we apply a softmax function, which normalizes these scores into a set of weights summed to one:
\begin{equation}
w_k = \frac{e^{p_k}}{\sum_{j=1}^{K} e^{p_j}},
\end{equation}
where \( w_k \) signifies the priority weight for camera \(k\). This method ensures that cameras which are more critical, either due to high coverage or due to lower delays, are given priority, thereby enhancing the decision-making capabilities and responsiveness of the edge analytics system.

\section{Problem Formulation}
\label{sec:Problem Formulation}
In this section, we establish the theoretical foundation for our PIB framework. We begin by detailing the IB analysis to determine the optimal balance between data compression and relevant information retention. Following this, we formulate the combinatorial multi-armed bandit (CMAB) problem to model the decision-making process of cameras in a distributed environment.

\subsection{Prioritized Information Bottleneck Analysis}
\label{sec:Prioritized Information Bottleneck Analysis}
In the context of information theory, the IB method seeks an optimal trade-off between the compression of an input variable $X$ and the preservation of relevant information about an output variable $Y$\cite{tishby2000information}. Throughout this paper, upper-case letters (e.g., \( X \), \( Y \), and \( Z \)) represent random variables, while lower-case letters (e.g., \( x \), \( y \), and \( z \)) denote their realizations. We formalize the input data from camera \( k \) as \( X^{(k)} \), and the target prediction as \( Y \), corresponding to the population in the dataset \( \mathcal{D} \). The goal is to encode \( X^{(k)} \) into a meaningful and concise representation \( Z^{(k)} \), which aligns with the hidden representation \( z^{(k)} \) that captures task-relevant features of multi-view content for prediction tasks. The classical IB problem can be formulated as a constrained optimization task:
\begin{equation}
\begin{aligned}
\max_{\Theta} \quad & \sum_{k=1}^{K}{I}\left(Z^{(k)};Y\right)\\
\text{s.t.} \quad  &{I}\left(X^{(k)};Z^{(k)}\right) \leq I_c,\quad (k=1,2,\cdots,K),
\end{aligned}
\end{equation}
where $I(Z^{(k)}, Y)$ denotes the mutual information between two random variables $ Z^{(k)} $ and $ Y $. $\Theta$ represents the set of all learnable parameters in the PIB framework, including ${\Theta_{M}}$ and the variational approximation in the following section. The mutual information is essentially a measure of the amount of information obtained about one random variable through the other random variable. $ I_c $ is the maximum permissible mutual information that $ Z^{(k)} $ can contain about $ X^{(k)}$. The objective is to ensure that $ Z^{(k)} $ captures the most relevant information about $X^{(k)}$ for predicting $ Y $ while remaining as concise as possible. By introducing a Lagrange multiplier\footnote{All Lagrange multipliers $\lambda$ are the same, and we only use trainable weight parameters to dynamically balance between accuracy and communication bottleneck.} $ \lambda $, the problem is equivalently expressed as:
\begin{equation}
\max_{\Theta} \quad R_{IB} =  \sum_{k=1}^{K}\left[{I}\left(Z^{(k)};Y\right) - \lambda {I}\left(X^{(k)};Z^{(k)}\right)\right],
\end{equation}
where $ R_{IB} $ represents the IB functional, balancing the compression of $ X^{(k)} $ against the necessity of accurately predicting $ Y $. Next, we extend the IB framework to a multi-camera setting by introducing priority weights to the mutual information terms, adapting the optimization problem as follows:
\begin{equation}\label{eq:IB}
\min_{\Theta} \quad  \sum_{k=1}^{K} \left[ - I_{w}\left(Z^{(k)};Y\right) + \lambda I_{w}\left(X^{(k)};Z^{(k)}\right) \right],
\end{equation}
where the weighted mutual information terms are defined as follows:
\begin{equation}
\begin{cases}
	I_w\left(Z^{(k)};Y\right)=w_k\cdot I\left(Z^{(k)};Y\right),\\
	I_w\left(X^{(k)};Z^{(k)}\right)=e^{w^0-w_k}\cdot I\left(X^{(k)};Z^{(k)}\right),\\
\end{cases}
\end{equation}
where the non-negative value \( w^0 \) represents the maximum allowable weight for \( w_k \). The first term with linear weights \( I_w\left(Z^{(k)};Y\right) = w_k I\left(Z^{(k)};Y\right) \) is the weighted mutual information between the compressed representation \( Z^{(k)} \) from camera \( k \) and the target \( Y \). This term can also be used to capture the semantic compression in raw data. The linear weighting with \( w_k \) ensures the influence of each camera is proportional to its priority weight. Higher \( w_k \) values increase the weight given to \( I\left(Z^{(k)};Y\right) \) in the objective function, emphasizing cameras that provide high-quality data for accurate target prediction.

{\color{black}{The second term with negative exponential weights \( I_w\left(X^{(k)};Z^{(k)}\right) = e^{(w^0-w_k)} \cdot I\left(X^{(k)};Z^{(k)}\right) \) denotes the mutual information between the original data \( X^{(k)} \) and its compressed form \( Z^{(k)} \), scaled by a negative exponential function of \( w_k \). This ensures an exponential decay in the influence of \( I\left(X^{(k)};Z^{(k)}\right) \) as \( w_k \) increases. Cameras with lower priority weights (lower \( w_k \)) undergo more aggressive data compression (as \( e^{w^0-w_k} \) is greater for smaller \( w_k \) values), optimizing bandwidth and storage usage without significantly affecting overall performance. In this paper, we use this type of weighting for the proof of concept study and will investigate more general weighting in the future.}}

\subsection{Combinatorial Multi-armed Bandit (CMAB) Problem}
\label{sec:CMAB Problem Formulation}
In dynamic environments with varying channel states and regions of interest (ROIs), ensuring high inference accuracy is challenging. The system must adaptively determine whether each camera should transmit its features and decide which edge server to use for transmission. Moreover, due to the edge server's limited bandwidth and computing capacity, each camera must decide if it should transmit directly to an edge server or use another edge server as a relay node before data fusion at the final edge server.

Accordingly, the problem can be formulated as a combinatorial multi-armed bandit (CMAB) problem. Each camera's connection establishment and edge server's connection establishment are base arms, and the collective actions of all agents constitute a super arm. Let \( a_k(t)\in \left\{ \mathcal{E}_{k,{s}}^{c\rightarrow e},\mathcal{E}_{{s}}^{e\rightarrow e_0} \right\} \) represent the action taken by camera \( k \) at time \( t \), where \(\mathcal{E}_{k,s}^{c \rightarrow e}\) denotes the connection between the \(k\)-th camera and the \(s\)-th edge server, and \(\mathcal{E}_{{s}}^{e \rightarrow e_0}\) denotes the connection between the \({s}\)-th edge server and the \({s}_0\)-th edge server for data fusion. The super arm is a subset of arms selected for the decision to transmit (the \(s\)-th edge server) and data fusion (the \(s_0\)-th edge server)\footnote{We omit ``$(t)$'' for simplicity in the definition of connection establishment.}. Dynamic channel state and ROI impact inference accuracy. This metric can be defined using the change in Multiple Object Detection Accuracy (MODA). MODA is calculated as \(M = 1 - \frac{FN + FP}{TP + FN}\), where \(TP\) denotes the number of correctly detected objects, \(FN\) means the number of missed detections, and \(FP\) is the number of false detections. Specifically, the gain in MODA from adding the \(k\)-th camera's feature to the ego camera's\footnote{The ego edge camera is the reference camera selected for data fusion. It typically has the highest number of detected moving objects (\(\chi_k\)) to provide the most comprehensive feature set for accurate object detection.} feature can be expressed as:
\begin{equation}
\label{eq:gain_in_MODA}
\Delta M_k = M_{\mathcal{C}_a \cup \{k\}} - M_{\mathcal{C}_a},
\end{equation}
where \(\mathcal{C}_a\) represents the set of cameras already selected, and \(M_{\mathcal{C}_a \cup \{k\}}\) represents the MODA score when the \(k\)-th camera is added to the set \(\mathcal{K}\).
To incorporate submodularity\footnote{The submodularity of $r_{a}(t)$ can be proven in Appendix C of \cite{fang2024pacp}.}, the reward function needs to reflect the diminishing returns property. Therefore, we define the reward function \(r_{\mathcal{K}}(t)\) as:
\begin{equation}
\label{eq:reward_function}
r_{a}(t) = \sum_{k \in \mathcal{K}} \Delta M_k,
\end{equation}
where \(\Delta M_k = M_{\mathcal{K} \cup \{k\}} - M_{\mathcal{K}}\), and \(\mathcal{K}\) is the set of cameras selected at time \(t\). The computational cost of multi-camera fusion and inference at the edge server \(s\) is denoted as \(\Psi^F_s\). The computational cost of simply forwarding features from one edge server to another is denoted as \(\Psi^T_s\). The remaining computing capacity of the \(k\)-th edge server is \(\Psi^R_k\). \(\mathcal{T}_{k,s}^{c \rightarrow e}\) denotes the latency of the transmission between the \(k\)-th camera and the \(s\)-th edge server, and \(\mathcal{T}_{{s}, {s}_0}^{e \rightarrow e}\) denotes the latency of the transmission between the aggregate node and the \(s\)-th edge server. Therefore, the CMAB problem can be formulated as:
\begin{small}
\begin{equation}
\label{eq:optimization_problem}
\begin{aligned}
&\underset{a_k(t)}{\max}\quad \sum_{t=1}^T{\mathbb{E}}[r_{a} (t)] \\
\quad \text{s.t.} \quad &(\ref{eq:optimization_problem}\textrm{a}): \quad {K}_{\min}\leq|\mathcal{K}|\leq{K}_{\max},\\
&(\ref{eq:optimization_problem}\textrm{b}): \quad \Psi^F\leq \Psi^R_{s_0},\\
&(\ref{eq:optimization_problem}\textrm{c}): \quad \sum_{s\in \mathcal{S}}{\mathcal{E}_{k,s}^{c\rightarrow e}}\cdot \mathcal{E}_{s,s_0}^{e\rightarrow e}=1, \forall k \in \mathcal{K},\\
&(\ref{eq:optimization_problem}\textrm{d}): \quad 0 \leq  \sum_{k \in \mathcal{K}} \mathcal{E}_{k, {s}}^{c \rightarrow e} \leq \mathcal{E}_{\max}^{c \rightarrow e},\forall s \in \mathcal{S},\\
&(\ref{eq:optimization_problem}\textrm{e}): \quad  \mathcal{E}_{k, {s}}^{c \rightarrow e} \mathcal{E}_{{s}}^{e \rightarrow e_0} \left( \mathcal{T}_{k, {s}}^{c \rightarrow e} + \mathcal{T}_{{s}, {s}_0}^{e \rightarrow e} \right) \leq \mathcal{T}^U, \forall k \in \mathcal{K},
\end{aligned}
\end{equation}
\end{small}
where (\ref{eq:optimization_problem}\textrm{a}) ensures the number of selected cameras falls within the specified range and $|\mathcal{K}|=\sum_{k\in \mathcal{K}}{\mathcal{E}_{k,s}^{c\rightarrow e}}$, (\ref{eq:optimization_problem}\textrm{b}) ensures that the remaining computing capacity of the aggregate server ($s_0$) chosen for inference is no less than the required capacity for fusion, (\ref{eq:optimization_problem}\textrm{c}) ensures that each camera uses a unique transmission connection, (\ref{eq:optimization_problem}\textrm{d}) ensures that the number of connections established by a single edge server does not exceed the maximum allowable connections, and (\ref{eq:optimization_problem}\textrm{e}) ensures that the total latency for any transmission path is within the allowable time limit $\mathcal{T}^U$ for all edge servers ${s} \in \mathcal{S}$. 

Solving the CMAB problem in multi-camera collaborative perception is challenging for traditional optimization methods due to: 1) \textbf{Dynamic environment}: Constantly changing channel states and ROIs make real-time adaptation difficult. 2) \textbf{Computational complexity}: The problem's combinatorial nature creates a massive solution space. 3) \textbf{Decentralized decision}: Independent yet collaborative decisions by multiple cameras and edge servers require a decentralized approach. Therefore, we employ distributed online learning techniques to address the CMAB problem in Sec. \ref{sec:Gate Mechanism Based on Distributed Online Learning}, allowing the system to learn and adapt dynamically, solve efficiently, and make decentralized decisions.


\section{Methodology}
\label{sec:Methodology}
In this section, we first introduce the overview of the proposed encoder/decoder architecture. Then, we propose the variational approximation method to reduce the computational complexity of estimating the mutual information during the minimization of Eq. (\ref{eq:IB}) in Sec. \ref{sec:Variational Approximation Method}. In Sec. \ref{sec: MFCM}, we design a multi-frame correlation model that utilizes variational approximation to capture the temporal correlation in video sequences. In Sec. \ref{sec:Network Loss Functions Derivation}, we derive the loss functions for the PIB-based encoder and decoder. Sec. \ref{sec:Gate Mechanism Based on Distributed Online Learning} proposes a gate mechanism based on distributed online learning to address the CMAB problem.

\subsection{Architecture Summary}
In this subsection, we outline the workflow of our PIB framework, designed for collaborative edge video analytics. As depicted in Fig. \ref{fig:encoder}, the process starts with each edge camera (denoted by $k$) capturing raw video data $X_t^{(k)}$ and extracting feature maps. These cameras utilize priority weights $w_k$ to optimize the balance between communication costs and perception accuracy, adapting to varying channel conditions. The extracted features are then compressed using entropy coding and sent as a bitstream to the edge server for further processing. On the server (see Fig. \ref{fig:decoder}), the video features are reconstructed using the shared parameters such as weights $w_k$ and the variational model parameters $q(Z_t^{(k)}|Z_{t-1}^{(k)}, ..., Z_{t-\tau}^{(k)})$. The server integrates these multi-view features to estimate $Y_t$, such as pedestrian occupancy and object detection. This approach leverages historical frame correlations through a multi-frame correlation model to enhance prediction accuracy. {\color{black}{The DNN architecture of the PIB framework is detailed in Appendix \ref{appendix:architecture}.}}

\begin{figure*}[t]
  \centering
  \begin{minipage}{0.48\textwidth}
    \centering
    \includegraphics[width=\textwidth]{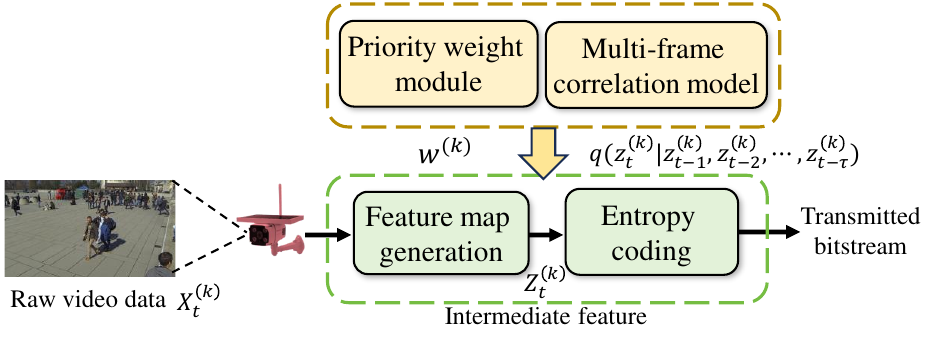}
    \caption{The procedure of video encoding.}
    \label{fig:encoder}
  \end{minipage}
  \hfill
  \begin{minipage}{0.48\textwidth}
    \centering
    \includegraphics[width=\textwidth]{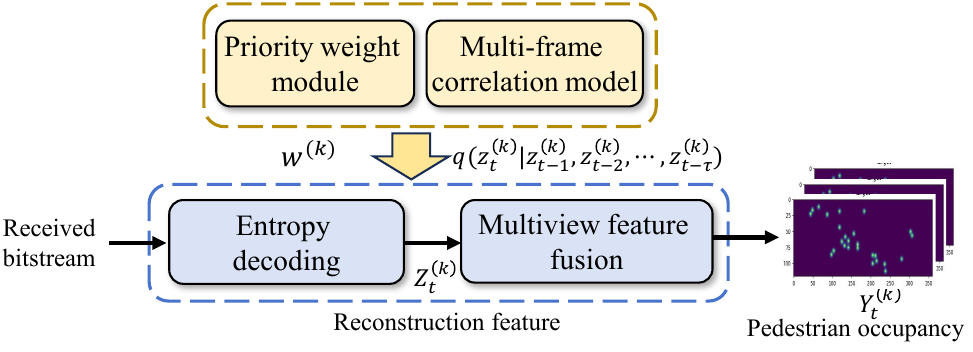}
    \caption{The procedure of video decoding.}
    \label{fig:decoder}
  \end{minipage}
\end{figure*}

\subsection{Variational Approximation Method}
\label{sec:Variational Approximation Method}
The objective function of information bottleneck in Eq. (\ref{eq:IB}) can be divided into two parts. The first part is $-\sum_{k=1}^K w_k\cdot I\left(Z^{(k)};Y\right)$, which denotes the quality of video reconstruction by decoding at an edge server. The second part is $\lambda \sum_{k=1}^{K} e^{w^0-w_k}\cdot I\left(X^{(k)};Z^{(k)}\right)$, which denotes the compression efficiency for feature extraction. As it has been shown in the way a decoder works, $p\left(Y|Z^{\left(k\right)}\right)$ can be any valid type of conditional distributions, but most often it is not feasible enough for straightforward calculation. Because of this complexity, it is highly challenging to directly compute the two mutual information components in Eq. (\ref{eq:IB}). 

As for the first part, we adopt the variational approach\cite{alemi2016deep}. This approach suggests that the decoder is part of a simpler group of distributions called $Q$. We then search for a distribution $q\left(Y|Z^{\left(k\right)};\Theta _{d}^{(k)}\right)$ within this group that is most similar to the best possible decoder distribution, using the KL-divergence to measure the closeness. $\Theta _{d}^{(k)}$ is a learnable parameter. Because computing the high-dimensional integrals in the posterior is infeasible, we substitute the optimal inference model with a variational approximation. Thus, we obtain the lower bound of $I_w\left( Z^{\left( k \right)};Y^{\left( k \right)} \right) =w_k\cdot I\left( Z^{\left( k \right)};Y^{\left( k \right)} \right) \geq w_k\left\{\mathbb{E}_{p\left( Y,Z \right)}\left[ \log q\left( Y^{\left( k \right)}|Z^{\left( k \right)};\Theta _{d}^{\left( k \right)} \right) \right] +H\left( Y^{\left( k \right)} \right)\right\}$, as established in Proposition \ref{proposition: lower_bound_mutual_information}.

\begin{proposition}
\label{proposition: lower_bound_mutual_information}
The probabilistic model of decoder \( p(Y|Z) \) maps a representation \( Z \in \mathbb{Z} \) into task inference \( Y \in \mathbb{Y} \). Let \( q(Y|Z) \) denote the variational approximation of decoder \( p(Y|Z) \). We can obtain
\begin{equation}\label{ineq: lower_bound_mutual_information}
I(Z;Y) \geq \mathbb{E}_{p(Y,Z)}[\log q(Y|Z)] + H(Y).
\end{equation}
\end{proposition}

\begin{proof}
We start with the standard definition of mutual information:
\begin{equation}
\begin{aligned}
I(Z;Y) = \mathbb{E}_{p(Y,Z)} \left[ \log \frac{p(Y,Z)}{p(Y)p(Z)} \right] = \mathbb{E}_{p(Y,Z)} \left[ \log \frac{p(Y|Z)}{p(Y)} \right],
\end{aligned}
\end{equation}
which utilizes the relationship \( p(Y, Z) = p(Y|Z)p(Z) \) to express the mutual information in terms of the ratio of the conditional probability to the marginal probability of \( Y \).

Introducing the Kullback-Leibler (KL) divergence, which measures how the distribution \( q(Y|Z) \) approximates the true distribution \( p(Y|Z) \), we have:
\begin{equation}
\label{eq:KL}
D_{KL} \left[ p(Y|Z) \parallel q(Y|Z) \right] = \mathbb{E}_{p(Y|Z)} \left[ \log \frac{p(Y|Z)}{q(Y|Z)} \right] \geq 0,
\end{equation}
where the KL divergence is always non-negative. This leads to:
\begin{equation}
\mathbb{E}_{p(Y|Z)} \left[ \log p(Y|Z) \right] \geq \mathbb{E}_{p(Y|Z)} \left[ \log q(Y|Z) \right],
\end{equation}
which can be simplified to:
\begin{equation}
\mathbb{E}_{p(Y,Z)} \left[ \log p(Y|Z) \right] \geq \mathbb{E}_{p(Y,Z)} \left[ \log q(Y|Z) \right].
\end{equation}

Therefore, we can derive the lower bound for the mutual information as follows:
\begin{equation}
\begin{aligned}\label{eq: IZY}
I(Z;Y) &= \mathbb{E}_{p(Y,Z)} \left[ \log \frac{p(Y|Z)}{p(Y)} \right] \\
       &\geq \mathbb{E}_{p(Y,Z)} [\log q(Y|Z)] - \mathbb{E}_{p(Y)} [\log p(Y)] \\
       &= \mathbb{E}_{p(Y,Z)} [\log q(Y|Z)] + H(Y),
\end{aligned}
\end{equation}
where \( H(Y) \) is the entropy of \( Y \), a constant that reflects the inherent uncertainty in \( Y \) independent of \( Z \).
{\hfill $\blacksquare$\par}
\end{proof}

To establish an upper bound for the term \(\lambda \sum_{k=1}^{K} e^{w_0-w_k} \cdot I(X^{(k)}; Z^{(k)})\) in the context of the complexity in directly minimizing it, we proceed as follows. Recognizing that \(H(Z^{(k)} | X^{(k)}) \ge 0\) from the properties of entropy, we obtain the inequality:
\begin{equation}
\begin{aligned}\label{eq:hz}
&\lambda \sum_{k=1}^K{I_w}\left( X^{(k)};Z^{(k)} \right) =\lambda \sum_{k=1}^K{\left[ \frac{H\left( Z^{(k)} \right) -H\left( Z^{(k)}|X^{(k)} \right)}{e^{w_k-w_0}} \right]} \\
&\le \lambda \sum_{k=1}^K{\frac{H\left( Z^{(k)} \right)}{e^{w_k-w_0}}}\le \lambda \sum_{k=1}^K{\frac{H\left( Z^{(k)}, V^{(k)} \right)}{e^{w_k-w_0}}},
\end{aligned}
\end{equation}
where we use the latent variables \(V^{(k)}\) as the side information to encode the quantized feature and we have used \(H(Z^{(k)}, V^{(k)}) \geq H(Z^{(k)})\). The joint entropy \(H(Z^{(k)}, V^{(k)})\) represents the communication cost, which is minimized when the joint entropy is minimized.

\begin{proposition}
\label{proposition: upper_bound_mutual_information}
The upper bound for the mutual information term in Eq. \ref{eq:IB} is given by:
\begin{equation}
\begin{aligned}\label{ineq:second_part_final}
I_w\left(X^{(k)};Z^{(k)}\right) \leq & \ \mathbb{E}_{p(Z^{(k)}, V^{(k)})} \left[ -\log q\left(Z^{(k)}|V^{(k)}; \Theta_{con}^{(k)}\right) \right. \\
& \left. \times q(V^{(k)}; \Theta_{l}^{(k)}) \right] e^{w_0-w_k},
\end{aligned}
\end{equation}
where \(q\left(Z^{(k)}|V^{(k)}; \Theta_{con}^{(k)}\right)\) is the variational distribution conditioned on the latent variables \(V^{(k)}\) with parameters \(\Theta_{con}^{(k)}\), and \(q(V^{(k)}; \Theta_{l}^{(k)})\) is the marginal variational distribution with parameters \(\Theta_{l}^{(k)}\).
\end{proposition}
\begin{proof}
The proof begins by recognizing that the joint entropy \(H(Z^{(k)}, V^{(k)})\) represents the communication cost, which is minimized when the joint entropy is minimized. The joint entropy can be expressed as the expectation over the logarithm of the ratio of the true joint distribution \(p(Z^{(k)}, V^{(k)})\) to the variational distribution \(q(Z^{(k)}|V^{(k)}; \Theta_{con}^{(k)}) \cdot q(V^{(k)}; \Theta_{l}^{(k)})\), where \(\Theta_{con}^{(k)}\) represents the learnable parameter:
\begin{equation}
\begin{aligned}\label{eq:joint_entropy}
H(Z^{(k)}, V^{(k)}) = &\mathbb{E}_{p(Z^{(k)}, V^{(k)})} \left[ -\log q(Z^{(k)}|V^{(k)}; \Theta_{con}^{(k)}) \right. \\
& \left. -\log q(V^{(k)}; \Theta_{l}^{(k)}) \right] - D_{KL}(p||{q}),
\end{aligned}
\end{equation}
where \(D_{KL}(p||{q})\) is the KL-divergence between the distribution of \(p = p(Z^{(k)}, V^{(k)})\) and \({q} = q(Z^{(k)},V^{(k)})\). The KL-divergence is non-negative, thus we have:
\begin{equation}
\begin{aligned}\label{eq:DKL}
D_{KL}\left[ p\left( Z^{(k)},V^{(k)} \right) \parallel q\left( Z^{(k)}|V^{(k)}; \Theta_{con}^{(k)} \right) q\left( V^{(k)}; \Theta_{l}^{(k)} \right) \right] \ge 0,
\end{aligned}
\end{equation}
Combining the joint entropy equation (\ref{eq:joint_entropy}) with inequality (\ref{eq:DKL}), we obtain:
\begin{equation}
\begin{aligned}
H(Z^{(k)}, V^{(k)}) \leq & \ \mathbb{E}_{p(Z^{(k)}, V^{(k)})} \left[ -\log q(Z^{(k)}|V^{(k)}; \Theta_{con}^{(k)}) \right. \\
& \left. \times q\left(V^{(k)}; \Theta_{l}^{(k)}\right) \right].
\end{aligned}
\end{equation}
Thus, we can substitute Ineq. (\ref{eq:DKL}) into Ineq. (\ref{eq:hz}) to obtain the result in Ineq. (\ref{ineq:second_part_final}).

{\hfill $\blacksquare$\par}
\end{proof}

It should be noted that the lower bound in Ineq. (\ref{eq: IZY}) and upper bound in Ineq. (\ref{ineq:second_part_final}) enables us to establish an upper limit on the objective function in minimization problem in (\ref{eq:IB}). This makes it easier to minimize with the corresponding loss function during network training, as discussed in \mbox{Sec. \ref{sec:Network Loss Functions Derivation}}.
\subsection{Multi-Frame Correlation Model}
\label{sec: MFCM}
Inspired by the previous work \cite{shao2023task}, PIB framework utilizes a multi-frame correlation model to leverage variational approximation to capture the temporal dynamics in video sequences. This approach utilizes the temporal redundancy across contiguous frames to model the conditional probability distribution effectively. Our model approximates the next feature in the sequence by considering the variational distribution \( q(Z_{t}^{(k)} | Z_{t-1}^{(k)}, ..., Z_{t-\tau}^{(k)}; \Theta_{\tau}^{(k)}) \), which can be modeled as a Gaussian distribution aimed at mimicking the true conditional distribution of the subsequent frame given the previous frames:
\begin{equation*}
q\left(Z_{t}^{(k)}|Z_{t-1}^{(k)},...,Z_{t-\tau}^{(k)};\Theta _{\tau}^{(k)}\right)=\mathcal{N} \left( \mu \left( \Theta _{\tau}^{(k)} \right) ,\sigma ^2\left(\Theta _{\tau}^{(k)} \right) \right) ,
\end{equation*}
where \( \mu \) and \( \sigma^2 \) are parametric functions of the preceding frames, encapsulating the temporal dependencies. These functions are modeled using a deep neural network with parameters \( \Theta_{\tau}^{(k)} \) learned from data. By optimizing the variational parameters, our model aims to closely match the true distribution, thus encoding the features more efficiently.

\subsection{Network Loss Functions Derivation}
\label{sec:Network Loss Functions Derivation}

In this subsection, we formulate our network loss functions to enhance the information transmission in a multi-camera scenario based on the priority-driven mechanism and the IB principle as discussed in Sec. \ref{sec: Priority Weight Formulation} and Sec. \ref{sec:Prioritized Information Bottleneck Analysis}.

Given the variability in channel quality and the occurrence of delays, we introduce the first loss function, \(\mathcal{L}^{(k)}_1\), designed to minimize the impact of unreliable data sources while maximizing inference accuracy. We also consider to improve the number of perceived moving objects (\({\chi}_{\text{norm}}\)). Thus, the loss function of the MLP network in Sec. \ref{sec: Priority Weight Formulation} is:
\begin{equation}\label{eq:L1}
\mathcal{L} _1=\sum_{k=1}^K{\left[ 1_{d_{\mathrm{norm},k}<\epsilon}\frac{\left( w_k-W_{\mathrm{target}} \right)^2}{\chi _{\mathrm{norm},k}}+1_{d_{\mathrm{norm},k}>\epsilon}\left( w_{k}^{2} \right) \right]},
\end{equation}
where \(\epsilon\) denotes a permissible delay that cannot result in errors in multi-view fusion, and  \(W_{\mathrm{target}}\) represents the target weight for a camera without excessive delay.
The second loss function \(\mathcal{L}_2\) aims to minimize the upper bound of the mutual information, following the inequalities derived in (\ref{eq: IZY}) and (\ref{ineq:second_part_final}). \(\mathcal{L}_2\) ensures efficient encoding while preserving essential information for accurate prediction:
\begin{equation}
\begin{aligned}\label{eq:L2}
\mathcal{L}_2 =& \sum_{k=1}^K{\underset{\mathrm{The}\ \mathrm{upper} \ \mathrm{bound} \ \mathrm{of} \ -I_w\left( Z^{(k)};Y \right)}{\underbrace{\mathbb{E} [-w_k\log q(Y|Z^{(k)} ;\Theta _{d}^{\left( k \right)})]}}} + \lambda  \cdot \min \bigg\{ R_{max}, \\
& \underset{\mathrm{The} \ \mathrm{upper} \ \mathrm{bound} \ \mathrm{of} \ I_w\left( X^{(k)};Z^{(k)} \right)}{\underbrace{\mathbb{E} \left[ -\log q(Z^{(k)}|V^{(k)};\Theta _{con}^{(k)})\cdot q(V^{(k)};\Theta _{l}^{(k)}) \right]e^{(w^0-w_k)} }} \bigg\}.
\end{aligned}
\end{equation}
The first term of $\mathcal{L}_2$ ignores $H(Y)$ in Ineq. (\ref{ineq: lower_bound_mutual_information}) because it is a constant. $R_{max}$ represents the penalty for the excessive communication cost of the variation approximation $q(Z^{(k)}|V^{(k)}; \Theta_{con}^{(k)}) \cdot q(V^{(k)}; \Theta_{l}^{(k)})$, which captures the degradation of training decoder $p(Y|Z^{(k)})$. In Sec. \ref{sec: MFCM}, the Multi-Frame Correlation Model leverages temporal dynamics, which is critical for sequential data processing in video analytics. The third loss function, \(\mathcal{L}^{(k)}_2\), is needed to minimize the KL divergence between the true distribution of frame sequences and the modeled variational distribution:
\begin{equation}\label{eq:L3}
\mathcal{L}_3 =  \sum_{k=1}^{K} D_{KL}\left[p(Z_t^{(k)} | Z_{<t}^{(k)}) || q(Z_t^{(k)} | Z_{<t}^{(k)};\Theta _{\tau}^{(k)})\right],
\end{equation}
where $Z_{<t}^{(k)}=(Z_{t-1}^{(k)},...,Z_{t-\tau}^{(k)})$. These loss functions collectively aim to optimize the trade-off between data transmission costs and perceptual accuracy, crucial for enhancing the performance of edge analytics in multi-camera systems. \mbox{ \ref{alg:variational-approximation}} introduces the detailed procedure of feature extraction and variational approximation.

\begin{algorithm}[t]
\caption{Training Procedures of the Feature Extraction and Variational Approximation}
\label{alg:variational-approximation}
\begin{algorithmic}[1]
\REQUIRE Training dataset, initialized parameters $\Theta_{d}^{(k)}$, $\Theta_{e}^{(k)}$, $\Theta_{con}^{(k)}$, $\Theta_{l}^{(k)}$, $\Theta_{\tau}^{(k)}$ for $k \in 1:K$, $\Theta_{M}$, $w^0$. 
\ENSURE Optimized parameters $\Theta_{e}^{(k)}$, $\Theta_{d}^{(k)}$, $\Theta_{con}^{(k)}$, $\Theta_{l}^{(k)}$ for $k \in 1:K$, and $\Theta_{M}$.
\REPEAT 
    \STATE Calculate the priority weights based on latency and sensing coverage of all cameras with parameter $\Theta_{M}$.
    \FOR{$k = 1$ to $K$}
        \STATE Extract the features by the feature extractor of camera $k$ with parameter $\Theta_{e}^{(k)}$.
        \STATE Compress the features based on the PIB framework with parameters $\Theta_{d}^{(k)}$, $\Theta_{e}^{(k)}$, $\Theta_{con}^{(k)}$, $\Theta_{l}^{(k)}$.
    \ENDFOR
    \STATE Compute the loss functions $\mathcal{L}_1$ and $\mathcal{L}_2$ in \mbox{Eqs. (\ref{eq:L1})-(\ref{eq:L2})}, respectively.
    \STATE Update parameters $\Theta_{d}^{(k)}$, $\Theta_{e}^{(k)}$, $\Theta_{con}^{(k)}$, $\Theta_{l}^{(k)}$ for $k \in 1:K$, and $\Theta_{M}$ through backpropagation.
\UNTIL Convergence of parameters $\Theta_{d}^{(k)}$, $\Theta_{e}^{(k)}$, $\Theta_{con}^{(k)}$, $\Theta_{l}^{(k)}$ for $k \in 1:K$, and $\Theta_{M}$.
\REPEAT 
    \FOR{$k = 1$ to $K$}
        \STATE Extract the features by the feature extractor of device $k$ with parameter $\Theta_{e}^{(k)}$.
        \STATE Compress the features based on the multi-frame correlation model with parameters $\Theta_{\tau}^{(k)}$.
        \STATE Compute the empirical estimation of the loss function $\mathcal{L}_3^{(k)}$ in Eq. (\ref{eq:L3}).
        \STATE Update parameters $\Theta_{\tau}^{(k)}$ through backpropagation.
    \ENDFOR
\UNTIL Convergence of parameters $\Theta_{\tau}^{(k)}$ for $k \in 1:K$.
\end{algorithmic}
\end{algorithm}

\subsection{Gate Mechanism Based on Distributed Online Learning}
\label{sec:Gate Mechanism Based on Distributed Online Learning}
The gate mechanism based on distributed online learning is designed to solve the CMAB problem (\ref{eq:optimization_problem}) formulated in Sec. \ref{sec:CMAB Problem Formulation}. In this subsection, we first introduce the intuitive ideas of gate mechanism. Then, we provide the details and explanation of the pseudocode. Finally, the evaluation of regret performance and communication cost for distributed execution is analyzed mathematically.
\subsubsection{Distributed Online Learning for CMAB Problem}
\label{sec:Distributed Online Learning for CMAB Problem}
Firstly, we propose a distributed Upper Confidence Bound (UCB) algorithm to address this problem, leveraging the independence of each camera agent to learn the optimal transmission strategy. This approach is particularly effective for managing the dynamic nature of the multi-camera network, where real-time channel quality and server load can significantly impact the overall system performance. Specifically, we assume that each arm represents the connection establishment between a camera and an edge server. The super arm $(\mathcal{E} _{k,{s}}^{c\rightarrow e},\mathcal{E} _{{s} ,{s} _0}^{e\rightarrow e_0})$ is the combination of these connections. The reward is defined based on the gain in MODA by adding the \(k\)-th camera's feature to the ego camera in \mbox{Eq. (\ref{eq:gain_in_MODA}).}

The intuitive idea behind using distributed UCB is to manage dynamic CSI and ROI efficiently. Each camera agent independently explores and exploits available edge servers based on local observations, making the system robust to changing network conditions. The algorithm has two phases: \textbf{exploration} and \textbf{exploitation}. In the exploration phase, each agent gathers information on potential rewards. In the exploitation phase, the agent selects the best action based on the UCB value, balancing exploration and exploitation under uncertainty. The UCB value for edge server \(s\) at time \(t\) is $\text{UCB}_{k,{s}}(t) = \hat{\mu}_{k,{s}}(t) + \alpha \sqrt{\frac{2 \ln t}{N_{k,{s}}(t)}}$, where \(\hat{\mu}_{k,{s}}(t)\) is the estimated reward, driving exploitation. The second term (\(\alpha \sqrt{\frac{2 \ln t}{N_{k,{s}}(t)}}\)) accounts for uncertainty, encouraging exploration\cite{10021290}. The key intuition is that actions that have been selected fewer times carry more uncertainty, so they are given a higher bonus to encourage exploration. Conversely, as an action is selected more often and its reward estimate becomes more reliable, the bonus decreases, leading to more exploitation of that action. The parameter \(\alpha\) balances how aggressively the algorithm explores uncertain actions versus exploiting known rewards. The square root component diminishes as more information is gathered, while the logarithmic term ensures the exploration bonus decreases slowly, encouraging exploration of less frequently chosen actions. Algorithm \ref{alg:DistUCBGate} provides the pseudocode for the proposed gate mechanism with distributed online learning.
\begin{algorithm}[t]
\caption{Gate Mechanism with Distributed Online Learning (DOL)}
\label{alg:DistUCBGate}
\begin{algorithmic}[1]
\STATE Initialize parameters: $\alpha$, $\beta$, $\gamma$
\FOR{each Camera $k = 1$ to $K$}
    \STATE Initialize reward estimates $\hat{\mu}_{k,{s}}(0) = 0$, action counts $N_{k,{s}}(0) = 0$, cumulative rewards $R_{k,{s}}(0) = 0$
\ENDFOR
\FOR{each time step $t = 1$ to $T$}
    \FOR{each Camera $k = 1$ to $K$}
        \STATE Update channel state $\text{CSI}_{k,{s}}(t)$ and edge server load $l_{{s}}(t)$ \label{alg-line: state update}
        \STATE Select edge server ${s}$ for fusion based on current state \label{alg-line: Select}
        \STATE Compute UCB value for each edge server ${s}$: \label{alg-line: ucb}
        \[
        \text{UCB}_{k,{s}}(t) = \hat{\mu}_{k,{s}}(t) + \alpha \sqrt{\frac{2 \ln t}{N_{k,{s}}(t)}}
        \]
        \STATE Select action $a_k(t) = \left\{ \mathcal{E} _{k,{s}}^{c\rightarrow e}, \mathcal{E} _{{s} ,{s} _0}^{e\rightarrow e_0} \right\}$ that maximizes UCB value
        \STATE Execute action $a_k(t)$ and observe reward $r_{k}(t)$ \label{alg-line: Execute action}
        
        \IF {constraints (\ref{eq:optimization_problem}\textrm{a}) or (\ref{eq:optimization_problem}\textrm{b}) or (\ref{eq:optimization_problem}\textrm{c}) or (\ref{eq:optimization_problem}\textrm{d}) or (\ref{eq:optimization_problem}\textrm{e}) are violated}
            \STATE Set $r_{k}(t) = 0$ \label{alg-line: set reward}
        \ENDIF

        \STATE Update action counts $N_{k,{s}}(t) = N_{k,{s}}(t-1) + 1$ \label{alg-line: update 1}
        \STATE Update cumulative rewards $R_{k,{s}}(t) = R_{k,{s}}(t-1) + r_{k}(t)$
        \STATE Update reward estimates $\hat{\mu}_{k,{s}}(t) = \frac{R_{k,{s}}(t)}{N_{k,{s}}(t)}$ \label{alg-line: update 3}   
    \ENDFOR
    \IF {Communication round is started} \label{alg-line: comm. start}
        \FOR{each Camera $k = 1$ to $K$}
            \STATE Aggregate rewards and action counts across agents
            \STATE Update global reward estimates $\hat{\mu}_{k,{s}}(t)$ for all edge servers ${s}$
        \ENDFOR
    \ENDIF \label{alg-line: comm. end}
\ENDFOR
\end{algorithmic}
\end{algorithm}


{\color{black}In Algorithm \ref{alg:DistUCBGate}, each camera agent independently updates its observed channel state and edge server load (Line \ref{alg-line: state update}). This enables decentralized learning of the optimal transmission strategy without centralized control. Specifically, in Lines \ref{alg-line: ucb}-\ref{alg-line: Execute action}, each agent computes the UCB value $\text{UCB}_{k,{s}}(t) = \hat{\mu}_{k,{s}}(t) + \alpha \sqrt{\frac{2 \ln t}{N_{k,{s}}(t)}}$ for each edge server \(s\), where \( \hat{\mu}_{k,{s}}(t) \) is the estimated reward and \( \alpha \) is used to adjust the balance between exploration and exploitation. By selecting the action with the highest UCB value, the agent optimally balances exploring the underutilized connections and exploiting the high-reward connections based on its local observations, guiding the system towards an optimal transmission strategy. Lines \ref{alg-line: update 1}-\ref{alg-line: update 3} handle updates to action counts, cumulative rewards, and reward estimates, refining the UCB values. Periodic communication rounds (Line \ref{alg-line: comm. start}) ensure consistency across agents by allowing each to share local reward estimates with a central server, which aggregates these data and updates global estimates for synchronization (Lines \ref{alg-line: comm. start}-\ref{alg-line: comm. end}). This process maintains synchronization across the distributed network while supporting scalability and adaptability.}

\subsubsection{Regret Analysis}
\label{sec:Regret Analysis}
The regret analysis reflects the efficiency and adaptability of the algorithm in optimizing network performance. A lower regret bound indicates that the algorithm performs close to the optimal strategy, ensuring high inference accuracy and efficient resource utilization despite the dynamic environment and varying network conditions. In the context of a multi-camera sensing network, the choices between different cameras and edge servers are interdependent. Furthermore, the connections between different cameras and edge servers exhibit heterogeneity, with each super arm having different distribution parameters. Therefore, we consider a CMAB problem with a non-identically distributed (non-i.i.d.) assumption and derive its regret upper bound in the following section.

The regret \( R_L(T) \) is an important metric to evaluate the performance of the online learning algorithm. The regret over a time horizon \( T \) is defined as the difference between the maximum expected reward obtainable by an optimal strategy and the expected reward obtained by the algorithm:
\begin{equation}
\begin{aligned}
R_L(T) = \sum_{t=1}^{T} \left( \max_{a} \mathbb{E}[r(a)] - \mathbb{E}[r(a_k(t))] \right),
\end{aligned}
\end{equation}
where \( a \) represents an action, and \( \mathbb{E}[r(a)] \) is the expected reward for action \( a \). To derive the regret bounds, we first establish a lemma using Bernstein's inequality in Lemma \ref{lemma:Bernstein}.
\begin{lemma}
\label{lemma:Bernstein}
(Bernstein's inequality) Let \( X_1, X_2, \ldots, X_T \) be independent random variables such that \( |X_t - \mathbb{E}[X_t]| \leq b \) almost surely. Then, for any \( \epsilon > 0 \),
\begin{equation}
\begin{aligned}
P\left( \left| \sum_{t=1}^{T} (X_t - \mathbb{E}[X_t]) \right| \geq \epsilon \right) \leq 2 \exp \left( -\frac{\epsilon^2}{2\sum_{t=1}^{T} \text{Var}(X_t) + \frac{2}{3}\epsilon} \right).
\end{aligned}
\end{equation}

\end{lemma}
\begin{proof}
Please refer to Sec. 2.8 in \cite{vershynin2018high}. {\hfill $\blacksquare$\par}
\end{proof}
\begin{theorem}
\label{theorem:regret_bound}
In a dynamic environment with non-identically distributed (non-i.i.d) rewards due to network heterogeneity, where the variance of the reward for arm \(k\) is denoted as \(\sigma_{r_k}^2\), the cumulative regret \(R(T)\) over \(T\) rounds of the distributed UCB algorithm is bounded by:
\begin{equation}
\begin{split}
R(T) \leq O\Bigg( & \left( \sqrt{2\sum_{k=1}^{\mathcal{K}^{\text{arm}}}\sigma_{r_k}^2} + 2\mathcal{K}^{\text{arm}}\sqrt{\frac{2}{a_N}} \right)\mathcal{K}^{\text{arm}} \sqrt{T \ln T} \\
& + \frac{2\mathcal{K}^{\text{arm}}}{3}\ln T  \Bigg),
\end{split}
\end{equation}
where \( \mathcal{K}^{\text{arm}}=\sum_{k=K_{\min}}^{K_{\max}}{C_{K}^{k}S^{k+1}} \) represents the maximum number of arms in the optimal super arm, \(\sigma_{r_k}\) is the standard deviation of the reward for arm \(k\), and \(a_N\) is an upper bound on the linear growth rate of the number of times UCB algorithm arms are selected over time\footnote{It indicates that \(N_k(t)\) follows a linear growth trend, i.e., \(N_k(t) \approx a_N t\).}. The non-i.i.d. nature of the rewards represents the heterogeneity in the network where different arms can have different reward distributions due to varying network conditions, processing capabilities, and data qualities.
\end{theorem}

\begin{proof}
We start by defining the reward for camera \(k\) transmitting to edge server \(s\) at time \(t\) as \(r_{k,s}(t)\). The mean reward for this transmission is denoted by \(\mu_{k,s}\), and the variance of the reward is \(\sigma_{k,s}^2\). The total variance in rewards across all camera-edge server pairs is represented by \(\sigma_r^2\), which accounts for variability due to both the dynamic channel state information (CSI) and the edge server load fluctuations. To derive the regret bound, we use Bernstein's inequality to bound the sum of rewards for each camera-edge server pair. For any \(\epsilon > 0\), Bernstein's inequality gives the following probability bound:
\begin{equation}
\begin{aligned}\label{eq:bernstein}
P\left( \left| \sum_{t=1}^{T} (r_{k,s}(t) - \mu_{k,s}) \right| \geq \epsilon \right) \leq 2 \exp\left( -\frac{\epsilon^2}{2 \sum_{t=1}^{T} \sigma_{k,s}^2 + \frac{2}{3}\epsilon} \right).
\end{aligned}
\end{equation}
Now, we set \(\epsilon\) as \(\epsilon = \sqrt{2T\sigma_{k,s}^2 \ln T} + \frac{2}{3}\ln T\) to account for the cumulative uncertainty over time \(T\). Substituting this into the right-hand side of Bernstein's inequality, we get:
\[
2\exp \left[ -\frac{2T\sigma _{k,s}^{2}\ln T+\frac{4}{3}\ln T\cdot \sqrt{2T\sigma _{k,s}^{2}\ln T}+\frac{4}{9}(\ln T)^2}{2T\sigma _{k,s}^{2}+\frac{2}{3}\sqrt{2T\sigma _{k,s}^{2}\ln T}+\frac{4}{9}\ln T} \right].
\]
As \(T\) increases, the dominant terms in the expression are \(2T\sigma_{k,s}^2 \ln T\) in both the numerator and the denominator. Therefore, we approximate the right-hand side as:
\begin{equation}
\begin{aligned}\label{eq: inter_2T}
2 \exp\left( -\frac{\epsilon^2}{2 \sum_{t=1}^{T} \sigma_{k,s}^2 + \frac{2}{3}\epsilon} \right) \approx 2 \exp(-\ln T) = \frac{2}{T}.
\end{aligned}
\end{equation}
Thus, combining Ineq. (\ref{eq:bernstein}) and Eq. (\ref{eq: inter_2T}), we obtain:
\begin{equation}
\begin{aligned}\label{eq:epsilon_bound}
P\left( \left| \sum_{t=1}^{T} (r_{k,s}(t) - \mu_{k,s}) \right| \geq \sqrt{2T\sigma_{k,s}^2 \ln T} + \frac{2}{3}\ln T \right) \leq \frac{2}{T}.
\end{aligned}
\end{equation}
It implies that, with high probability, when \(T\) is sufficiently large, we have:
\begin{equation}
\begin{aligned}\label{eq:high_prob}
\left| \sum_{t=1}^{T} (r_{k,s}(t) - \mu_{k,s}) \right| \leq \sqrt{2T\sigma_{k,s}^2 \ln T} + \frac{2}{3}\ln T.
\end{aligned}
\end{equation}
Therefore, the regret for each camera-edge server pair, denoted by arm \( (k,s) \), can be bounded as:
\begin{equation}
\begin{aligned}\label{eq:regret_each_arm}
R_{k,s}(T) \leq \sum_{t=1}^{T} (\mu_{k,s}^* - \mu_{k,s}) + \sqrt{2T\sigma_{k,s}^2 \ln T} + \frac{2}{3}\ln T,
\end{aligned}
\end{equation}
where \(\sum_{t=1}^{T} (\mu_{k,s}^* - \mu_{k,s})\) represents the regret due to not always selecting the optimal arm \( (k,s) \), while the term \(\sqrt{2T\sigma_{k,s}^2 \ln T} + \frac{2}{3}\ln T\) captures the uncertainty in reward estimation. Since the UCB algorithm selects the arm \( (k,s) \) that maximizes the UCB value, we have:
\begin{equation}
\begin{aligned}
\mu_{k,s}^* \leq \text{UCB}_{k,s}(t) = \hat{\mu}_{k,s}(t) + \sqrt{\frac{2 \ln t}{N_{k,s}(t)}},
\end{aligned}
\end{equation}
where \(N_{k,s}(t)\) is the number of times the camera-edge server pair \( (k,s) \) has been selected up to time \(t\). This implies that:
\begin{equation}
\begin{aligned}
\mu_{k,s}^* - \mu_{k,s} \leq \left( \hat{\mu}_{k,s}(t) - \mu_{k,s} \right) + \sqrt{\frac{2 \ln t}{N_{k,s}(t)}}.
\end{aligned}
\end{equation}
Therefore, the upper bound on the cumulative loss (i.e., regret) for arm \( (k,s) \) can be expressed as:
\begin{equation}
\begin{aligned}\label{eq:regret_decompose}
\sum_{t=1}^{T} (\mu_{k,s}^* - \mu_{k,s}) \leq \sum_{t=1}^{T} \left( \hat{\mu}_{k,s}(t) - \mu_{k,s} + \sqrt{\frac{2 \ln t}{N_{k,s}(t)}} \right).
\end{aligned}
\end{equation}
Since \(\hat{\mu}_{k,s}(t)\) is an unbiased estimate of \(\mu_{k,s}\), the expected value of \(\hat{\mu}_{k,s}(t) - \mu_{k,s}\) is zero. Thus, the regret is mainly determined by the term \(\sqrt{\frac{2 \ln t}{N_{k,s}(t)}}\). We approximate the cumulative sum of this term by using an integral, given that \(N_{k,s}(t)\) is assumed to grow linearly with time, i.e., \(N_{k,s}(t) \approx a_N t\), where \(a_N\) is a constant. Under this assumption, we have:
\begin{equation}
\begin{aligned}
\int_{1}^{T} \frac{1}{\sqrt{N_{k,s}(t)}} dt \approx \int_{1}^{T} \frac{1}{\sqrt{a_N t}} dt = \frac{2\sqrt{T} - 2}{\sqrt{a_N}}.
\end{aligned}
\end{equation}
Thus, the regret for a single arm \( (k,s) \) can be bounded as:
\begin{equation}
\begin{aligned}
R_{k,s}(T) \leq \sqrt{2T\sigma_{k,s}^2 \ln T} + \frac{2}{3}\ln T + \left( 2\sqrt{T}-2 \right) \sqrt{\frac{2\ln T}{a_N}}.
\end{aligned}
\end{equation}

To obtain the total regret \( R(T) \), we sum the regret over all camera-edge server pairs in the set of arms \( \mathcal{K} \):
\begin{small}
\begin{equation}
\begin{aligned}\label{eq:total_regret}
R(T) \leq \sum_{(k,s) \in \mathcal{K}} \left( \sqrt{2T\sigma_{k,s}^2 \ln T} + \frac{2}{3}\ln T + \left( 2\sqrt{T}-2 \right) \sqrt{\frac{2\ln T}{a_N}} \right).
\end{aligned}
\end{equation}
\end{small}

Then, we sum the bias and variance terms across all arms. For the bias term, we have:
\begin{equation}
\begin{aligned}\label{eq:bias_term}
\sum_{(k,s) \in \mathcal{K}} \sqrt{2T\sigma_{k,s}^2 \ln T} \leq \mathcal{K}^{\text{arm}} \sqrt{2T\sigma_r^2 \ln T},
\end{aligned}
\end{equation}
where \( \mathcal{K}^{\text{arm}} \) is the number of arms in the optimal super arm, and \(\sigma_r^2\) is the maximum variance among all arms. Since $O$ notation represents an asymptotic upper bound, the negative term $- 2\mathcal{K}^{\text{arm}}\sqrt{\frac{2\ln T}{a_N}}$ is omitted to maintain the validity of the bound. Therefore, when \( T \) is sufficiently large, the overall regret bound can be expressed as $R(T)  \leq O\left( \left( \sqrt{2\sigma_r^2} + 2\sqrt{\frac{2}{a_N}} \right)\mathcal{K}^{\text{arm}} \sqrt{T \ln T} + \frac{2\mathcal{K}^{\text{arm}}}{3}\ln T \right).$
Thus, the cumulative regret \(R(T)\) for the distributed UCB algorithm is bounded by the sum of the regret from all camera-edge server pairs, ensuring that the regret grows sub-linearly with respect to \(T\). 
{\hfill $\blacksquare$\par}
\end{proof}

\begin{proposition}
\label{proposition: complexity}
It is assumed that there are $N$ camera agents and $T$ time steps. The computational complexity of the proposed DOL method in Algorithm \ref{alg:DistUCBGate} for each camera agent at each time step is $O(\mathcal{K}^{\text{arm}}\log \mathcal{K}^{\text{arm}})$. The overall time complexity is $O(TN\mathcal{K}^{\text{arm}}\log \mathcal{K}^{\text{arm}})$.
\end{proposition}
\begin{proof}
Each camera agent \(k\) computes the UCB value for each edge server \({s}\) at each time step \(t\). This involves updating the channel state and edge server load, computing the UCB values, selecting the optimal action, and updating the reward estimates. As for each camera agent, the complexity of updating the channel state and edge server load for each camera agent is \(O(\mathcal{K}^{\text{arm}})\). Moreover, the complexity of computing the UCB value for each edge server is \(O(\mathcal{K}^{\text{arm}})\). The complexity of selecting the action that maximizes the UCB value is \(O(\mathcal{K}^{\text{arm}} \log \mathcal{K}^{\text{arm}})\). Thus, the complexity for each camera agent at each time step is \(O(\mathcal{K}^{\text{arm}} \log \mathcal{K}^{\text{arm}})\). Given that there are \(N\) camera agents and \(T\) time steps, the overall time complexity is $O(TN\mathcal{K}^{\text{arm}}\log \mathcal{K}^{\text{arm}})$.
{\hfill $\blacksquare$\par}
\end{proof}

\begin{proposition}
\label{proposition: comm cost}
Assuming there are \(\mathcal{X}\) communication rounds over \(T\) time steps, the total communication cost of Algorithm \ref{alg:DistUCBGate} is $O\left( \mathcal{X}N\mathcal{K}^{\text{arm}} \right) $.
\end{proposition}
\begin{proof}
Each communication round involves local communication between each camera agent and the central edge server, as well as the global aggregation and update phases. 1) \textit{Local Communication}: Each camera agent \(k\) communicates its local reward estimates \(\hat{\mu}_{k,{s}}(t)\) and action counts \(N_{k,{s}}(t)\) for each edge server \({s}\) to the central edge server. The communication cost for each agent per round is \(O(\mathcal{K}^{\text{arm}})\). Given \(N\) agents, the total local communication cost per round is \(O(N\mathcal{K}^{\text{arm}})\). 2) \textit{Global Aggregation}: The central edge server aggregates the information from all \(N\) camera agents. The complexity of aggregating the information is \(O(N\mathcal{K}^{\text{arm}})\). 3) \textit{Global Update}: The central server then broadcasts the updated global reward estimates to all \(N\) agents. The communication cost for broadcasting is \(O(N\mathcal{K}^{\text{arm}})\). Assuming there are \(\mathcal{X}\) communication rounds over \(T\) time steps, the total communication cost is $O(\mathcal{X}N\mathcal{K}^{\text{arm}})$.
{\hfill $\blacksquare$\par}
\end{proof}

\begin{figure*}[t]
  \centering
  \includegraphics[width=1.02\textwidth]{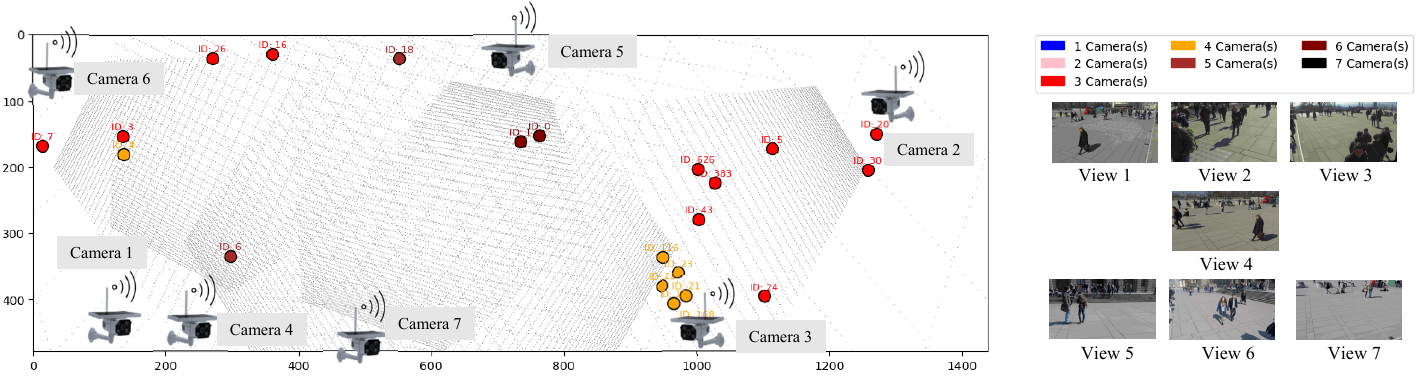}
  \caption{The visualization of the edge video system. \textbf{Left}: We use contour lines to display the perception coverage of different cameras. The small dots in the grid represent pedestrians, with different colors of the dots indicating the number of cameras covering each pedestrian. It can be observed that areas closer to the perception center of cameras are covered by more cameras. \textbf{Right}: The visualization of the raw video data and the legend for different numbers of covered cameras.}
  \label{fig:Sensing-result}
\end{figure*}
\begin{figure}[t]
  \centering
  \subfigure[Pedestrian perception result using only \textbf{single camera}.]{
    \includegraphics[width=0.48\textwidth]{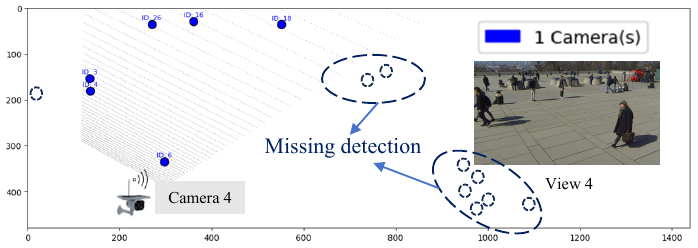}\label{fig:single_perception}
  }
  \subfigure[Pedestrian perception result using \textbf{two collaborative cameras.}]{
    \includegraphics[width=0.48\textwidth]{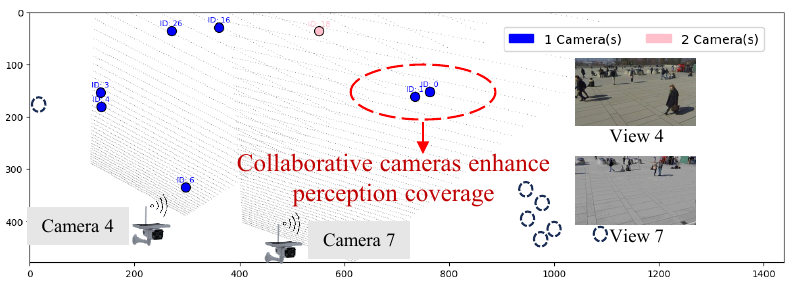}\label{fig:collaborative_perception}
  }
  \caption{Comparison of single and collaborative perception results. Fig. \ref{fig:single_perception} shows the detection using only Camera 4. Fig. \ref{fig:collaborative_perception} demonstrates the enhanced detection capability achieved through the collaboration between Camera 4 and Camera 7.}
  \label{fig:perception_comparison}
  \vspace{-3mm}
\end{figure}
\section{Performance Evaluation}
\label{sec:Performance Evaluation}
\subsection{Simulation Setup}
\begin{figure}[t]
  \centering
  \includegraphics[width=0.50\textwidth]{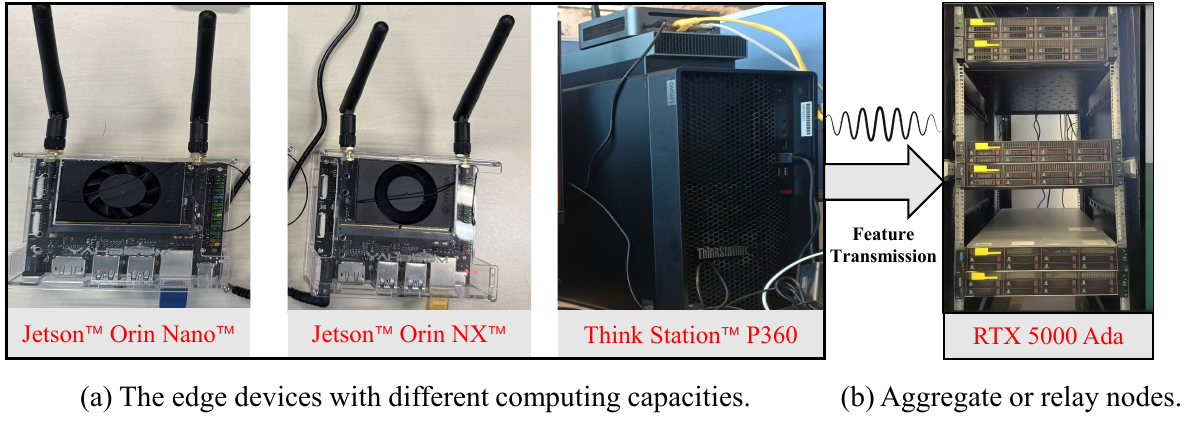}
  \caption{{\color{black}{Real-world Hardware Testbed.}}}
  \label{fig:Hardware}
\end{figure}
We set up simulations to evaluate our PIB framework, aiming at predicting pedestrian occupancy in urban settings using multiple cameras. These simulations replicate a city environment, with variables like signal frequency and device density affecting the outcomes.

Our simulations use a 2.4 GHz operating frequency, a path loss exponent of 3.5, and a shadowing deviation of 8 dB. Devices emit an interference power of 0.1 Watts, with densities ranging from 10 to 100 devices per 100 square meters, allowing us to test different levels of congestion. The bandwidth is set to 2 MHz, with cameras located at about 200 meters from the edge server. We employ the \textit{Wildtrack} dataset from EPFL, which features high-resolution images from seven cameras located in a public area, capturing unscripted pedestrian movements\cite{chavdarova2018wildtrack}. This dataset provides 400 frames per camera at 2 frames per second, documenting over 40,000 bounding boxes that highlight individual movements across more than 300 pedestrians. {\color{black}{As shown in Fig. \ref{fig:Hardware}, our experimental setup features a practical hardware testbed that includes three distinct edge devices: NVIDIA Jetson™ Orin Nano™ 4GB, NVIDIA Jetson™ Orin NX™ 16GB, and ThinkStation™ P360. The edge devices collaboratively interact with edge servers equipped with RTX 5000 Ada GPUs for efficient video decoding.}}
Our code will be made available at \href{https://github.com/fangzr/PIB-Prioritized-Information-Bottleneck-Framework}{github.com/fangzr/PIB-Prioritized-Information-Bottleneck-Framework}.

The primary measure we use is MODA, which assesses the system’s ability to accurately detect pedestrians based on missed and false detections. We also look at the rate-performance tradeoff to understand how communication overhead affects system performance. For comparative analysis, we consider five baselines, including video coding and image coding:
\begin{itemize}
    \item \textbf{TOCOM-TEM}\cite{shao2023task}: A task-oriented communication framework utilizing a temporal entropy model for edge video analytics. It employs the deterministic IB principle to extract and transmit compact, task-relevant features, integrating spatial-temporal data on the server for improved inference accuracy.
    \item \textbf{JPEG}\cite{wallace1992jpeg}: A widely used image compression standard employing lossy compression algorithms to reduce image data size, commonly used to decrease communication load in networked camera systems.
    \item \textbf{H.265}\cite{bossen2012hevc}: Also known as High Efficiency Video Coding (HEVC) or MPEG-H Part 2, which offers up to 50\% better data compression than its predecessor H.264 (MPEG-4 Part 10), while maintaining the same video quality, crucial for efficient data transmission in high-density camera networks.
    \item \textbf{H.264}\cite{H264}: Known as Advanced Video Coding (AVC) or MPEG-4 Part 10, which significantly enhances video compression efficiency, allowing high-quality video transmission at lower bit rates.
    \item \textbf{AV1}\cite{han2021technical}: AOMedia Video 1 (AV1) is an open, royalty-free video coding format developed by the Alliance for Open Media (AOMedia), designed to succeed VP9 with improved compression efficiency. AV1 outperforms existing codecs like H.264 and H.265, making it ideal for online video applications.
\end{itemize}

\begin{figure*}[t]
  \centering
  \subfigure[Communication bottleneck vs MODA.]{
    \includegraphics[width=5.5cm]{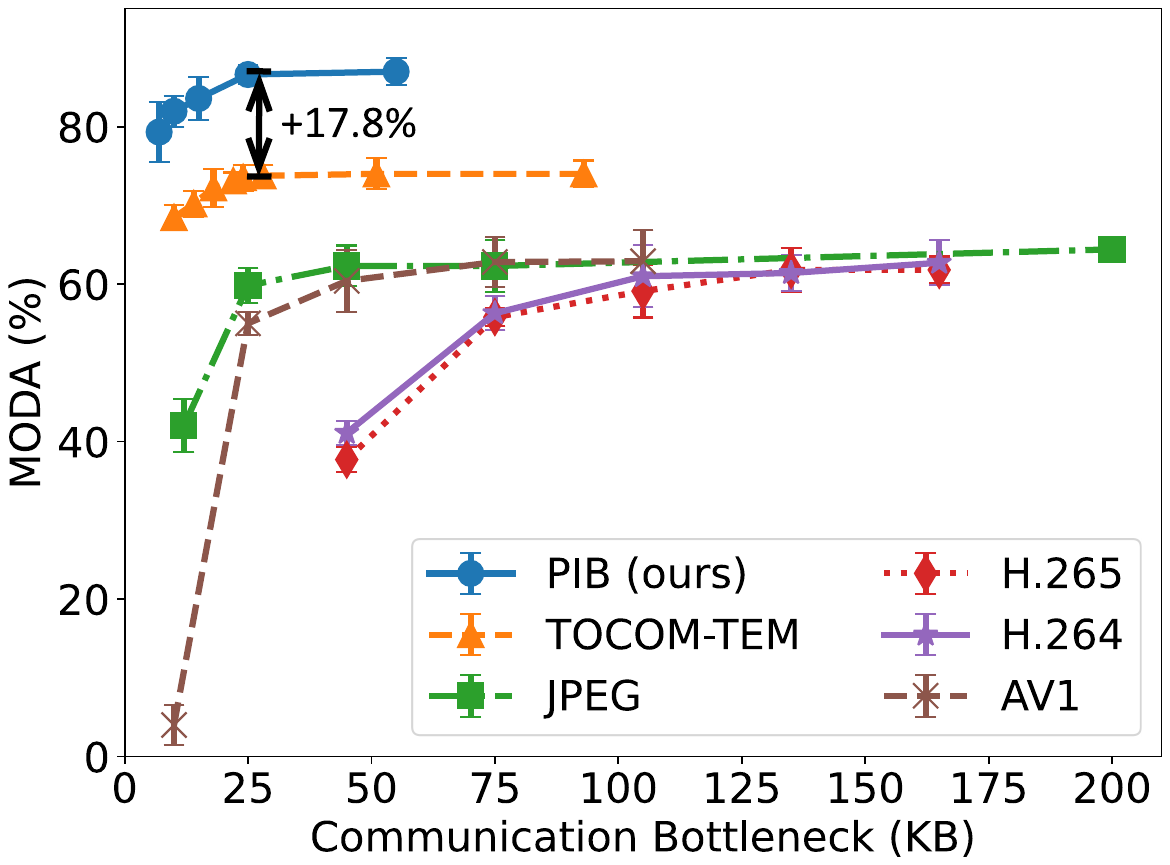}
    \label{fig:CBMA}
  }
  \subfigure[{\color{black}{Communication bottleneck vs MODP.}}]{
    \includegraphics[width=5.5cm]{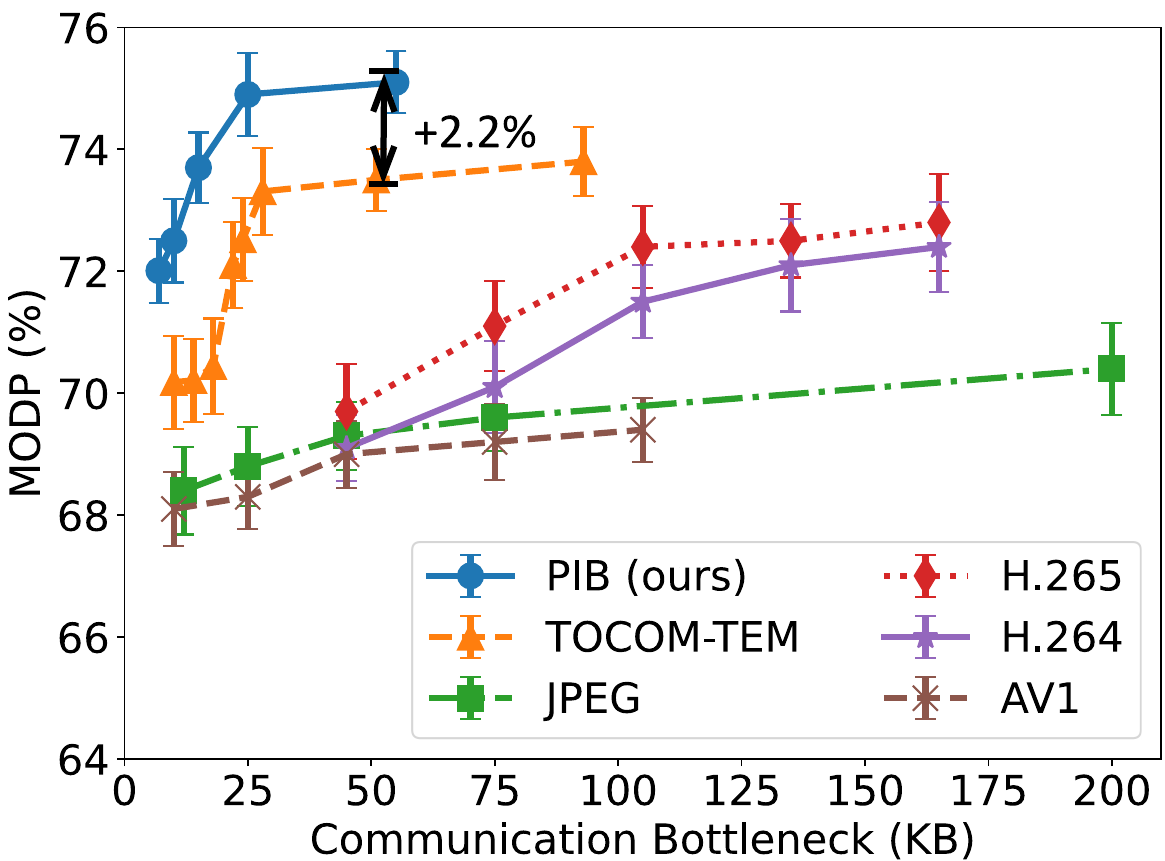}
    \label{fig:CBMP}
  }
  \subfigure[Number of delayed cameras vs MODA.]{
    \includegraphics[width=5.5cm]{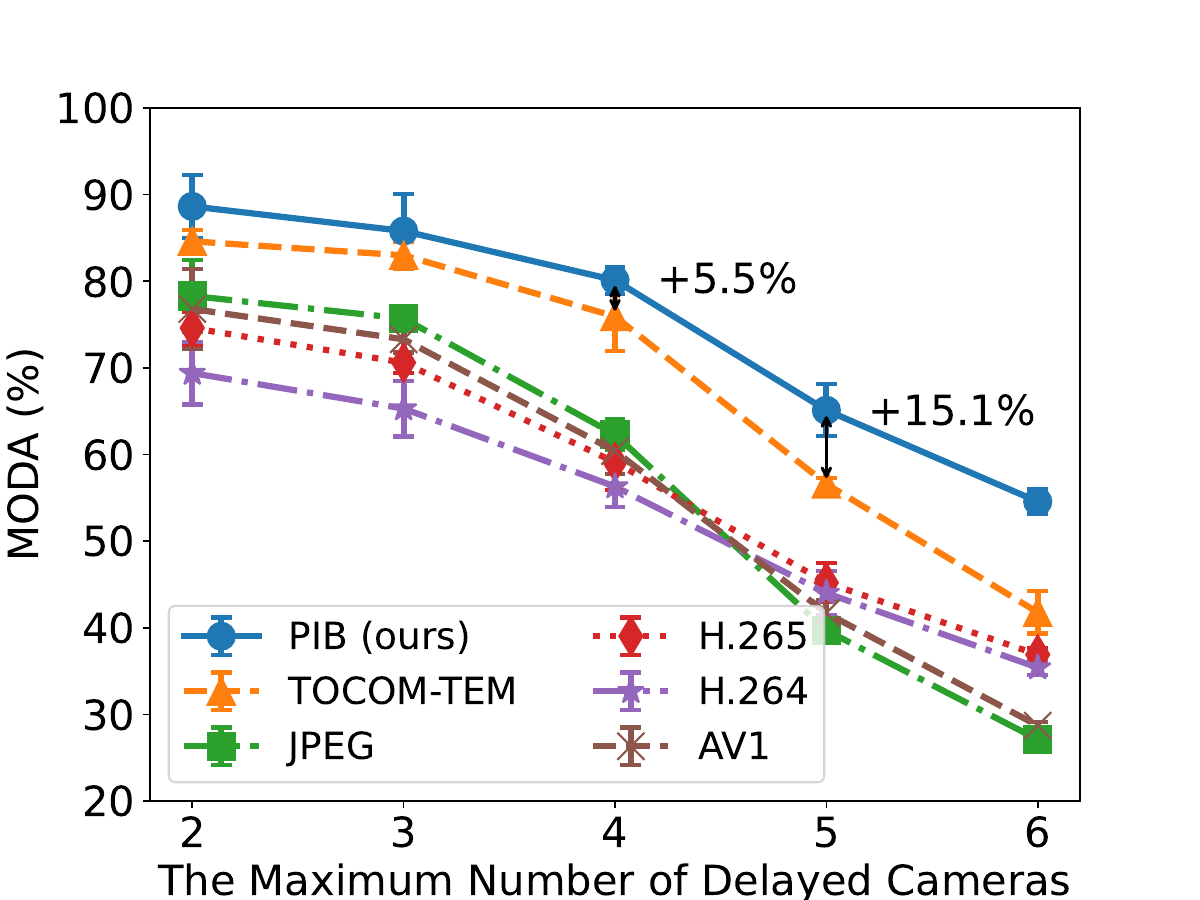}
    \label{fig:num-vs-moda}
  }
  \caption{Impact of communication bottlenecks and delayed cameras on perception accuracy.}
  \label{fig:combined-figure}
  \vspace{-3mm}
\end{figure*}

In the simulation study, we examine the effectiveness of multiple camera systems in forecasting pedestrian presence. Unlike a single-camera configuration, this method minimizes obstructions commonly found in crowded locations by integrating perspectives from various angles. Fig. \ref{fig:Sensing-result} demonstrates our experimental setup, where seven wireless edge cameras jointly perceive a 12m×36m area quantized into a 480×1440 grid using a resolution of 2.5 $\text{cm}^2$. We use contour lines to display the camera's perception range and the resolution of coordinates within that range. The denser the lines, the closer the perceived target is to the camera, and the higher the perception accuracy. Additionally, to clearly show the coverage of pedestrians at different positions by edge cameras, different colors represent the number of cameras covering each pedestrian. It can be observed that pedestrians in different locations have different probabilities of being detected, which will also affect the priority selection of cameras. Fig. \ref{fig:single_perception} shows the perception results using a single camera (the 4th edge camera). The dashed circles represent pedestrians that are missing detection. It is evident that the perception range of a single camera is limited to its own angle and coverage area, resulting in numerous missing detections. In Fig. \ref{fig:collaborative_perception}, we let the 4th and 7th edge cameras collaborate with each other. It can be observed that the collaboration enhances perception coverage, though there are still several pedestrians not detected compared to the results from seven edge cameras. This highlights the improved but still limited capability of collaborative perception with only two cameras, indicating the necessity for a higher number of cameras to achieve comprehensive coverage and accurate pedestrian detection\footnote{Our demo is available at the url: \href{https://github.com/fangzr/PIB-Prioritized-Information-Bottleneck-Framework}{github.com/fangzr/PIB-Prioritized-Information-Bottleneck-Framework}}.

To evaluate the impact of communication bottlenecks and delayed cameras on perception accuracy, we present in Fig. \ref{fig:CBMA}--\ref{fig:num-vs-moda} the relationships between communication constraints and the perception accuracy. Nevertheless, the benefit of collaborative perception is accompanied by excessive communication overhead. The communication bottleneck refers to network capacity constraints that prevent real-time data transmission, causing frame latency. This issue is prevalent in UDP-based wireless streaming systems, where high throughput often results in out-of-order or delayed frames due to varying channel quality and jitter. Moreover, different coding schemes cause varying delays under dynamic channel conditions, misaligning data fusion due to channel quality and jitter. Therefore, in order to evaluate how latency differences affect perception accuracy, we set communication bottleneck constraints. In our experiments, we use MODA (Multiple Object Detection Accuracy) and MODP (Multiple Object Detection Precision) to assess coding efficiency and robustness.

In Fig. \ref{fig:CBMA}, PIB exhibits higher MODA across different communication bottlenecks compared to five baselines by more than 17.8\%. This is due to PIB's strategic multi-view feature fusion, informed by channel quality and priority-based ROI selection. PIB prioritizes the shared information to mitigate delays that could degrade multi-camera perception accuracy. Interestingly, JPEG outperforms video coding schemes like H.265 and AV1 in our experiments, due to the low FPS of 2 used for video transmission, which does not leverage motion prediction advantages. AV1 performs well due to its high compression efficiency compared to H.264 and H.265. Fig. \ref{fig:CBMP} shows that PIB achieves higher MODP performance compared to three other baselines. The results indicate that MODP is less affected by latency because it measures the precision of detection without considering missed detections, whereas MODA is more impacted as it accounts for both missed and false detections. 

Fig. \ref{fig:num-vs-moda} depicts the performance rates of different compression techniques in a multi-view scenario in terms of the number of delayed cameras. Our proposed PIB method and TOCOM-TEM, both utilizing multi-frame correlation models, effectively reduce redundancy across multiple frames, achieving superior MODA at equivalent compression rates. PIB, in particular, employs a prioritized IB framework, enabling an adaptive balance between compression rate and collaborative sensing accuracy, optimizing MODA across various channel conditions. It is worth noting that the impact on collaborative perception MODA can be ignored in scenarios with fewer delayed cameras (<3). However, as channel conditions worsen and more cameras experience frame delays due to failing to meet communication bottleneck constraints, the performance significantly degrades.

\begin{figure}[t]
  \centering
  \includegraphics[width=0.47\textwidth]{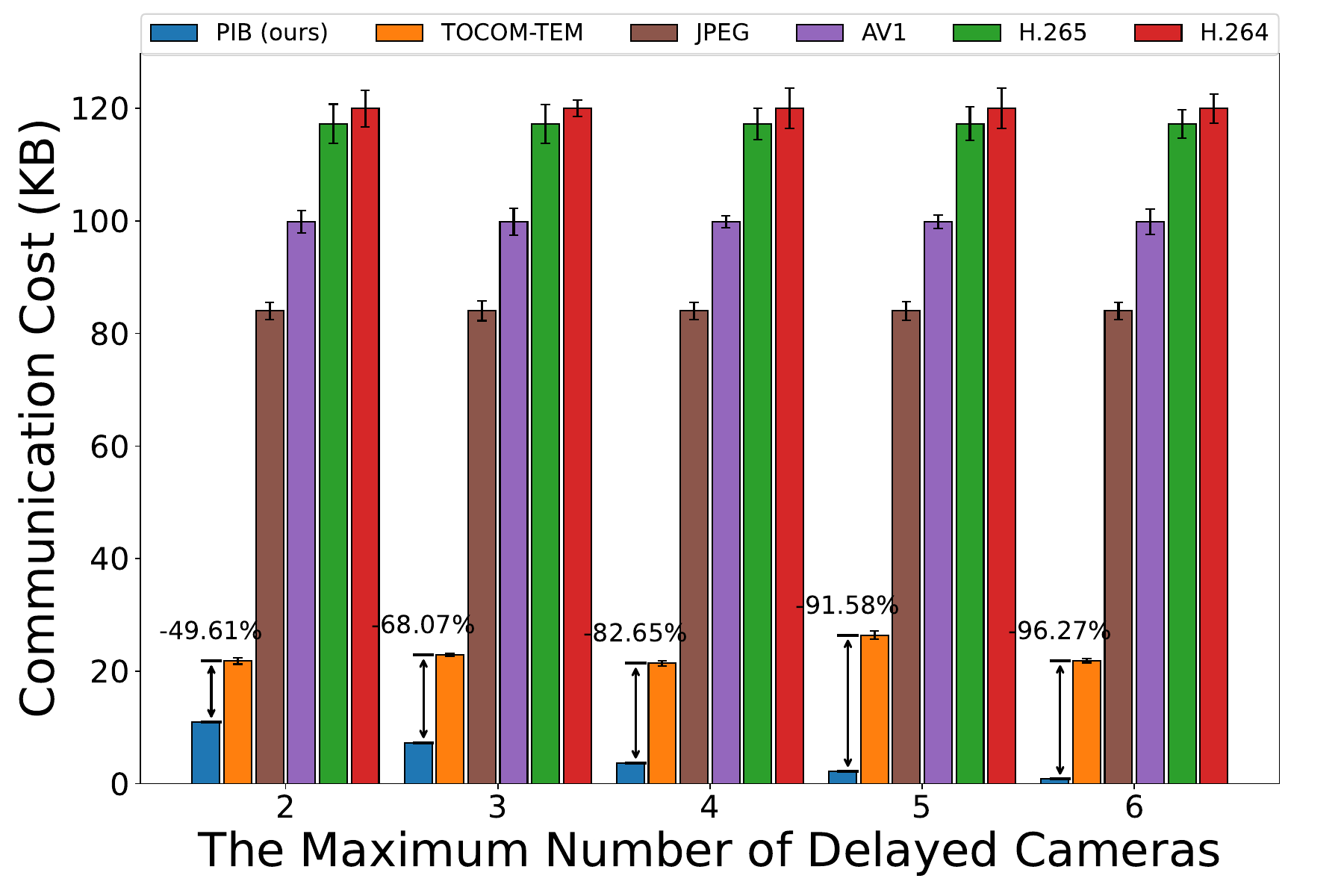}
  \caption{{\color{black}{Delayed cameras vs communication cost.}}}
  \label{fig:num-vs-communication cost}
  \vspace{-3mm}
\end{figure}
 
In Fig. \ref{fig:num-vs-communication cost}, we analyze the impact of the number of delayed cameras on the communication cost\footnote{The communication cost of a method is the average size of each frame. The instantaneous streaming rate is equal to the communication cost multiplied by the frames per second (fps).} for various algorithms. The PIB algorithm demonstrates a significant reduction in communication costs as the number of delayed cameras increases. When the number of delayed cameras equals 4, PIB, utilizing a gate mechanism based on a distributed UCB algorithm, effectively filters out useless streaming data, greatly reducing communication costs. Compared to TOCOM-TEM, PIB achieves an impressive 82.8\% decrease in communication costs. This efficiency is due to the algorithm's priority mechanism, which adeptly assigns weights and filters out adverse information caused by delays. Consequently, PIB prioritizes the transmission of high-quality features from cameras with more accurate occupancy predictions. For a fair comparison, baselines are selected at their highest MODA with the minimum communication cost data. Due to the use of an information bottleneck framework, PIB extracts only task-related features, resulting in a significantly reduced compression rate compared to five compression baselines.

\begin{figure}[t]
  \centering
  \subfigure[Streaming packet size for PIB and five baselines over time.]{
    \includegraphics[width=0.47\textwidth]{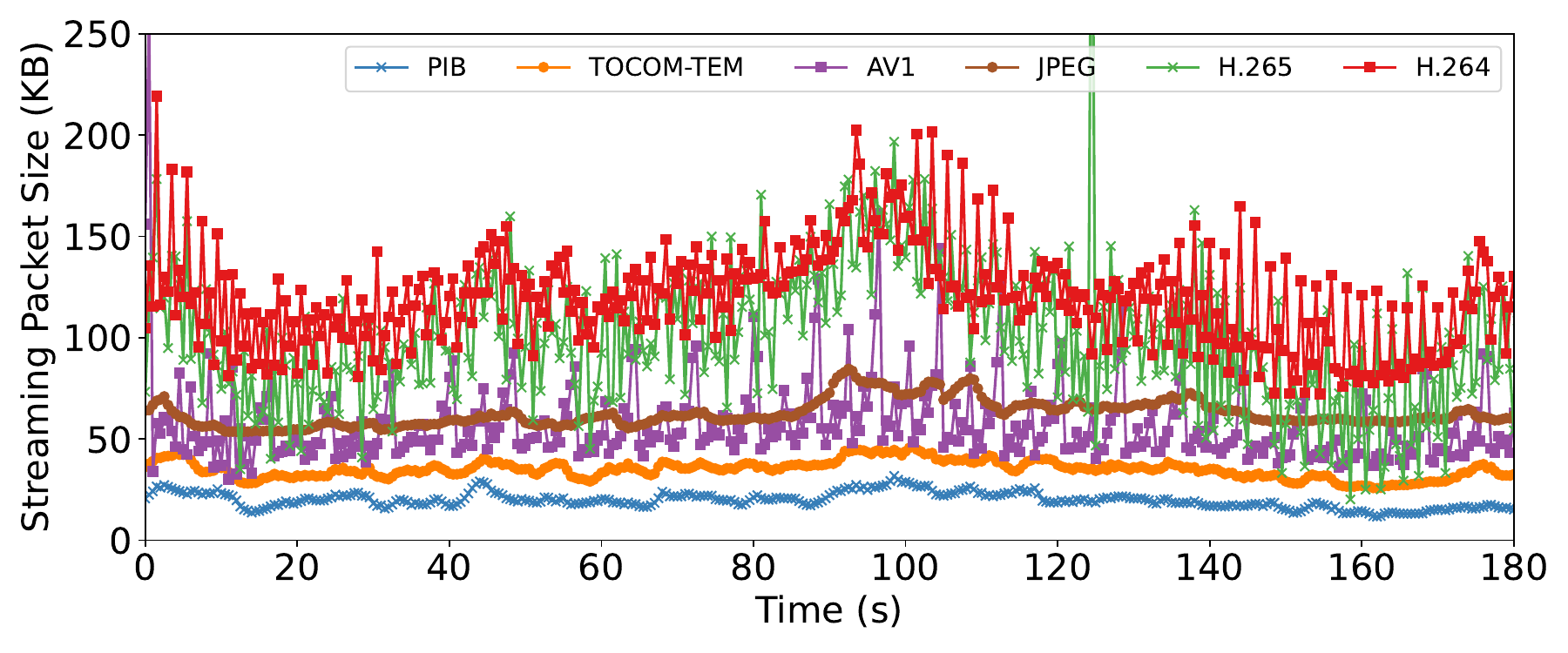}\label{fig:time_vs_bitrate}
  }
  \subfigure[CDF for streaming packet sizes.]{
    \includegraphics[width=0.47\textwidth]{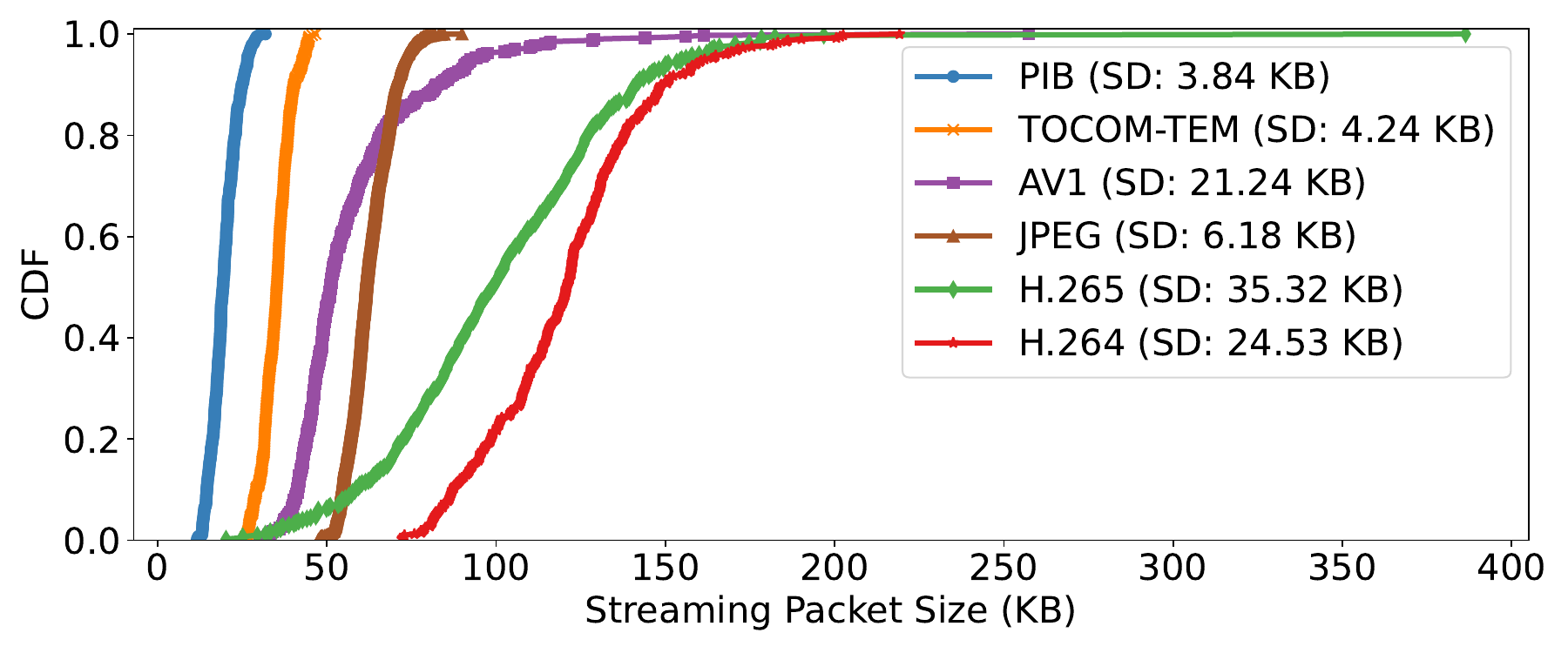}\label{fig:cdf}
  }
  \caption{(a) Streaming packet sizes for various compression algorithms over time slots. (b) Cumulative Distribution Functions (CDF) of the streaming packet sizes for different methods.}
  \label{fig:streaming_packet_size_comparison}
  \vspace{-3mm}
\end{figure}

Fig.~\ref{fig:streaming_packet_size_comparison} presents the streaming packet sizes and their cumulative distribution functions (CDF) for various compression algorithms. Fig.~\ref{fig:time_vs_bitrate} illustrates the streaming packet sizes for PIB, TOCOM-TEM, AV1, JPEG, H.265, and H.264 over a duration of three minutes. All encoding methods were evaluated at their highest MODA with minimal communication costs. PIB consistently exhibits the smallest packet sizes, followed by TOCOM-TEM, indicating superior transmission efficiency. Additionally, PIB and TOCOM-TEM demonstrate less variability in packet sizes compared to AV1, enhancing transmission robustness under adverse channel conditions. JPEG compression yields smaller and more stable packet sizes than H.264 and H.265, likely due to the limited transmission rate of 2 fps restricting the efficiency of video codecs. Fig.~\ref{fig:cdf} shows the CDF of streaming packet sizes for all algorithms. The standard deviation (SD) for each method is calculated as $\text{SD} = \sqrt{\frac{\sum_{i=1}^{n} (\text{Packet Size}_i - \text{Mean})^2}{n}}$. A lower SD indicates improved transmission robustness by reducing jitter and minimizing buffer requirements. PIB has the lowest SD (3.84 KB), followed by TOCOM-TEM (4.24 KB) and JPEG (6.18 KB). The other baseline methods exhibit higher SD values, underscoring PIB's advantage in minimizing both transmission requirements and packet size variability.

\begin{table}[t]
    \centering
    \caption{Impact of the Number of Fusion Cameras on Collaborative Perception Accuracy and Communication Cost.}
    \begin{tabular}{@{}>{\centering\arraybackslash}p{1.5cm}>{\centering\arraybackslash}p{1.8cm}>{\centering\arraybackslash}p{2.2cm}>{\centering\arraybackslash}p{2.2cm}@{}}
        \toprule
        \textbf{Number} & \textbf{Comm. Cost} & \textbf{MODA (\%)} & \textbf{MODP (\%)} \\ 
        \midrule
        \textbf{1} & 2.19 KB & 65.11 (+17.99\%) & 71.53 (+2.59\%) \\
        \textbf{2} & 3.68 KB & 78.09 (+19.93\%) & 72.71 (+1.65\%) \\
        \textbf{3} & 7.29 KB & 84.99 (+8.85\%) & 72.92 (+0.29\%) \\
        \textbf{4} & 10.98 KB & 88.03 (+3.57\%) & 74.23 (+1.80\%) \\
        \textbf{5} & 15.84 KB & 88.64 (+0.69\%) & 75.15 (+1.24\%) \\
        \textbf{6} & 17.68 KB & 88.76 (+0.14\%) & 75.80 (+0.87\%) \\
        \midrule
        \textbf{No Fusion} & \textbf{0.82 KB} & \textbf{55.17} & \textbf{69.72} \\
        \bottomrule
    \end{tabular}
    \label{tab:fusion_camera}
\end{table}
\begin{figure}[t]
  \centering
  \includegraphics[width=0.47\textwidth]{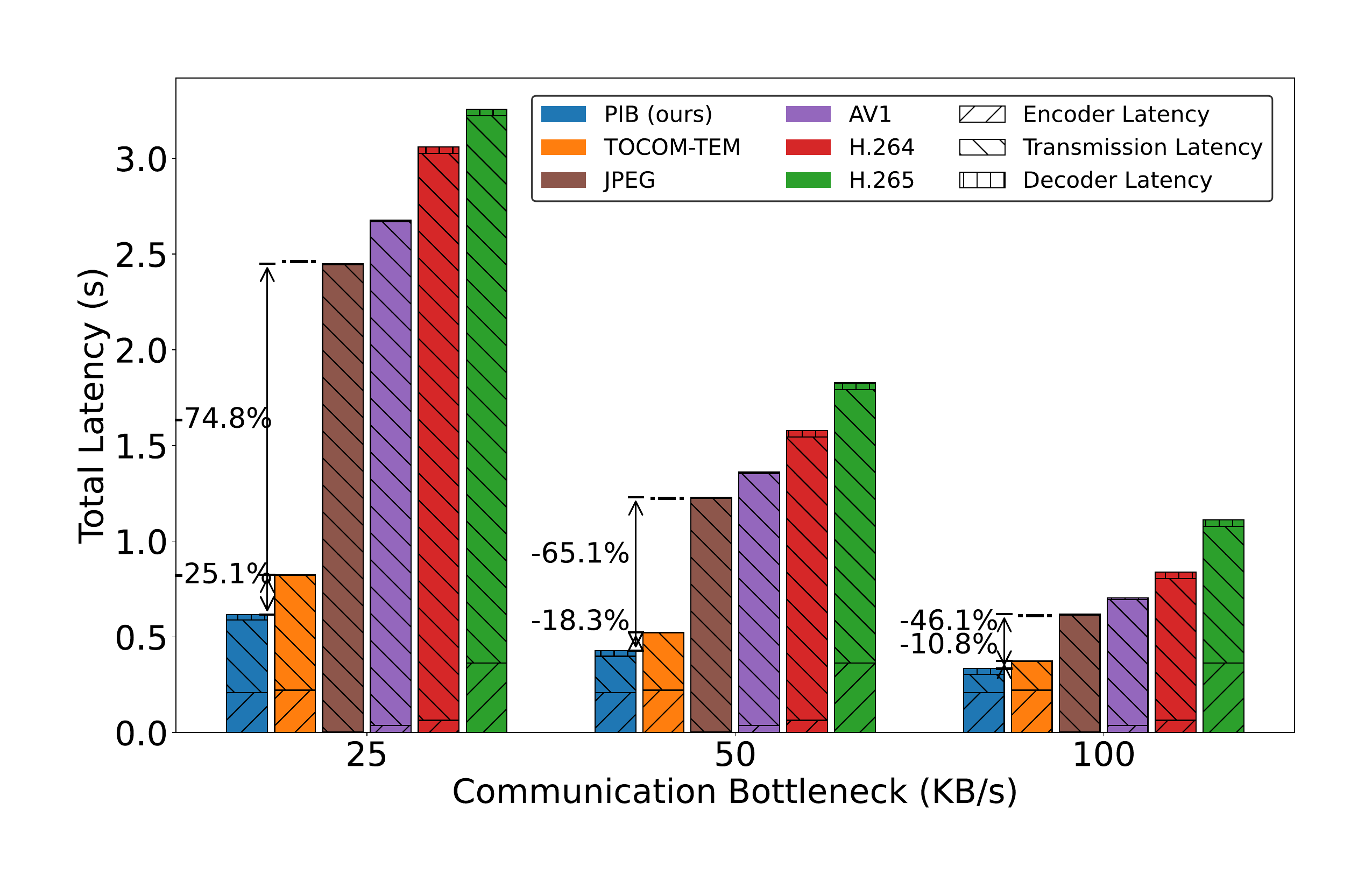}
  \caption{\color{black}{{Communication bottleneck vs latency.}}}
  \label{fig:Communication bottleneck vs latency}
  \vspace{-3mm}
\end{figure}

\begin{table}[t]
\centering
\color{black} 
\caption{\color{black}{Encoder Latency Across Different Platforms.}}
\label{tab:encoding_latency}
\begin{tabular}{|l|c|c|c|}
\hline
\backslashbox{\textbf{Phase}}{\textbf{Platform}} & \textbf{Nano (ms)} & \textbf{Orin NX (ms)} & \textbf{P360 (ms)} \\ \hline
Feature map generation     & 755.32±69.32 & 227.54±2.65 & 37.49±0.90 \\ \hline
Entropy coding      & 10.83±3.51   & 1.79±0.75   & 0.40±0.11  \\ \hline
\textbf{Total encoder latency}      & \textbf{766.15±70.55} & \textbf{229.34±2.67} & \textbf{37.80±0.94} \\ \hline
\end{tabular}
\end{table}

{\color{black}{Our priority-based mechanism selects the camera with the most targets within the RoI for the highest transmission priority. Collaboration priority is thus determined by the target count in each camera's perception area. As shown in Table \ref{tab:fusion_camera}, just 0.82 KB of perception data achieves a MODA accuracy of 55.17\% and MODP of 69.72\%, highlighting significant redundancy among edge cameras. Adding more cameras initially improves perception significantly but offers diminishing returns as communication costs to the edge server increase. In Fig. \ref{fig:Communication bottleneck vs latency}, we show the relationship between communication bottleneck and total latency for different algorithms. By leveraging a priority information bottleneck framework and the UCB algorithm to reduce redundancy, our PIB, despite slightly higher encoding latency than the traditional video codecs, can achieve much lower transmission latency due to its efficient compression. Under a 25 KB/s bottleneck, our PIB reduces latency by 25.1\% over TOCOM-TEM and 74.8\% over JPEG. At 50 KB/s, our PIB outperforms TOCOM-TEM by 18.3\% and JPEG by 65.1\%, respectively. At 100 KB/s, PIB achieves 10.8\% and 46.1\% lower latency than TOCOM-TEM and JPEG, respectively. The encoding latency results of our PIB in different edge devices are presented in Table \ref{tab:encoding_latency}. It can be observed that the feature map generation phase dominates the overall encoding latency, while the entropy coding phase contributes a negligible amount of time. Furthermore, edge devices with higher computating capacity exhibit significantly lower encoding latency.} }

\begin{figure}[t]
  \centering
  \includegraphics[width=0.47\textwidth]{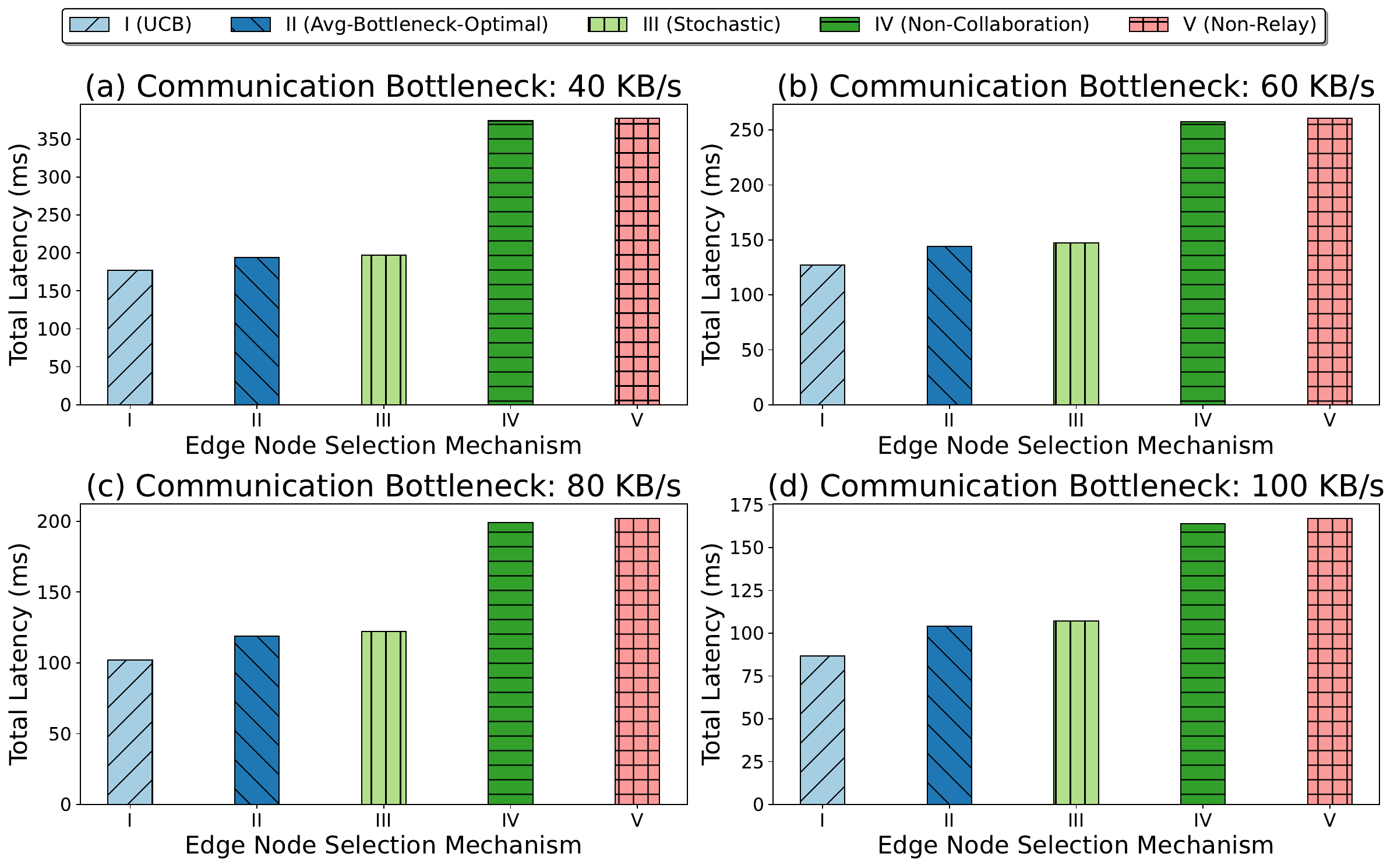}
  \caption{Impact of communication bottleneck on total latency for different edge server selection mechanisms.}
  \label{fig:Communication rate vs latency with diff selection method}
  \vspace{-3mm}
\end{figure}
Fig. \ref{fig:Communication rate vs latency with diff selection method} demonstrates the effectiveness of the proposed Gate Mechanism Based on Distributed Online Learning (Sec. \ref{sec:Gate Mechanism Based on Distributed Online Learning}). This figure evaluates total latency, defined as the sum of inference, relay, and transmission latency, excluding encoder latency, under different communication bottlenecks for various edge node selection mechanisms. Four baselines are used: \textit{Avg-Bottleneck-Optimal} (exhaustive search for highest average bottleneck), \textit{Stochastic} (random selection of relay and fusion nodes), \textit{Non-Collaboration} (single edge server for fusion), and \textit{Non-Relay} (lowest load edge server). The UCB method consistently achieves the lowest total latency, adapting efficiently to edge server load and channel conditions with minimal overhead, thereby optimizing collaborative selection of edge servers. Fig. \ref{fig:edge-servers-vs-latency} illustrates the impact of different numbers of edge servers on the latency of multi-camera collaborative sensing data transmission and inference under varying communication bottlenecks. The results indicate that as the number of edge servers increases, the overall average latency of the cameras significantly decreases. This is because the communication bottleneck is comparable in magnitude to the size of the intermediate representations transmitted by the cameras. Therefore, increasing the number of edge servers markedly reduces latency, showcasing the effectiveness of adding more edge servers in enhancing system performance.

\begin{figure}[t]
  \centering
  \includegraphics[width=0.50\textwidth]{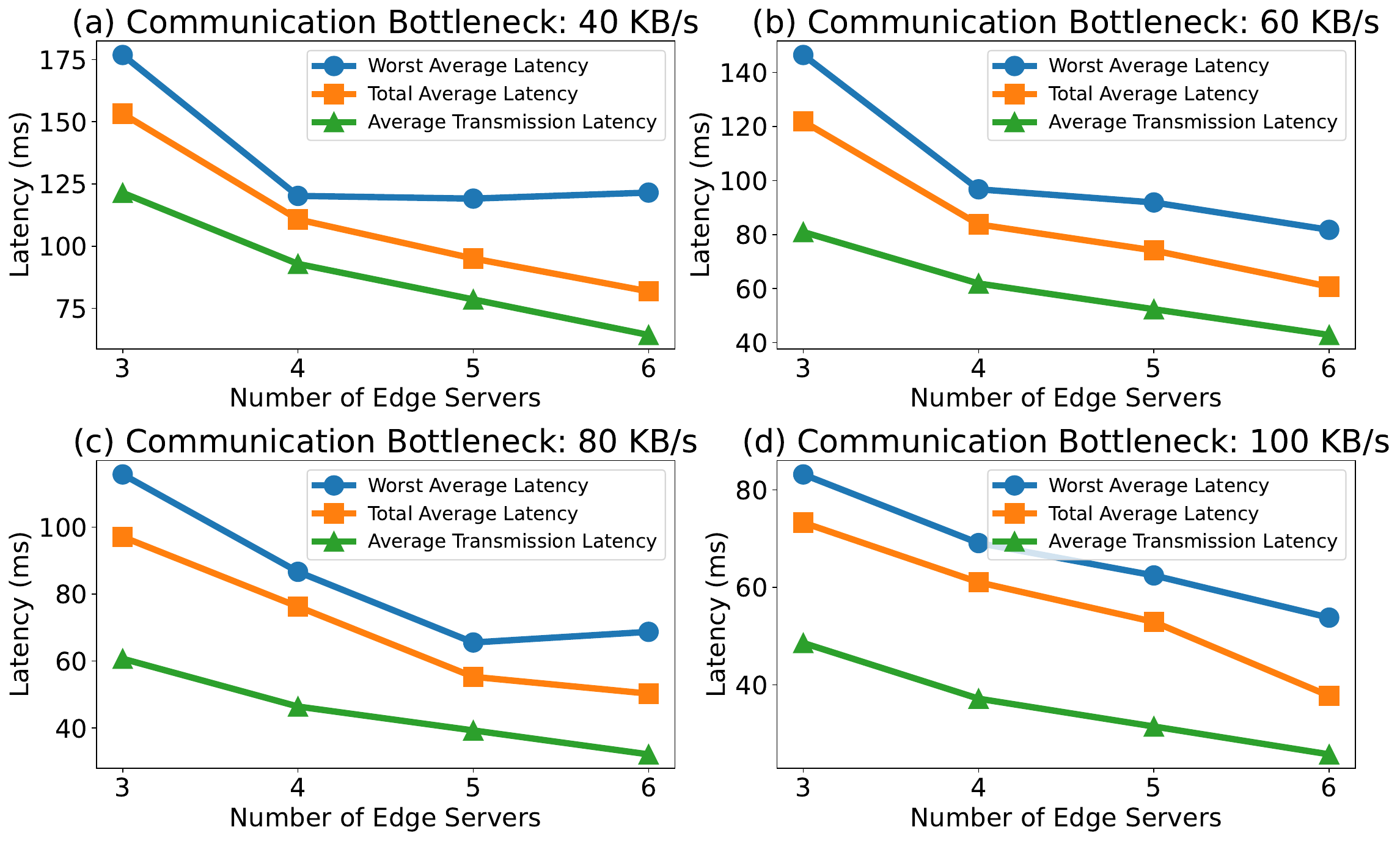}
  \caption{The number of edge servers vs various latencies under different communication bottlenecks.}
  \label{fig:edge-servers-vs-latency}
  \vspace{-3mm}
\end{figure}

\section{Conclusion}
In this paper, we have proposed the Prioritized Information Bottleneck (PIB) framework as a robust solution for collaborative edge video analytics. Our contributions are two-fold. First, we have developed a prioritized inference mechanism to intelligently determine the importance of different cameras' FOVs, effectively addressing the constraints imposed by channel capacity and data redundancy. Second, the PIB framework showcases its effectiveness by notably decreasing communication overhead and improving tracking accuracy without requiring video reconstruction at the edge server. Extensive real-world experiments show that: PIB not only surpasses the performance of conventional methods like TOCOM-TEM, JPEG, H.264, H.265, and AV1 with a marked improvement of up to 17.8\% in MODA but also achieves a considerable reduction in communication costs by 82.65\%, while retaining low latency and high-quality multi-view sensory data processing under less favorable channel conditions.

\ifCLASSOPTIONcaptionsoff
  \newpage
\fi

\bibliographystyle{IEEEtran}
\bibliography{ref}

\begin{IEEEbiography}[{\includegraphics[width=1in,height=1.25in,clip,keepaspectratio]{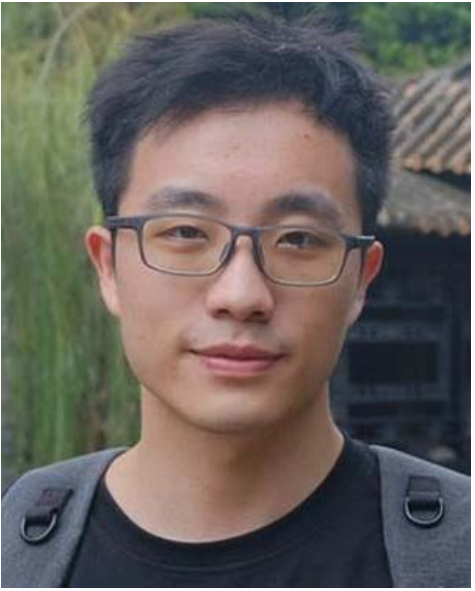}}] {Zhengru Fang} (S'20) received his B.S. degree (Hons.) in electronics and information engineering from the Huazhong University of Science and Technology (HUST), Wuhan, China, in 2019 and received his M.S. degree (Hons.) from Tsinghua University, Beijing, China, in 2022. Currently, he is pursuing his PhD degree in the Department of Computer Science at City University of Hong Kong. His research interests include collaborative perception, V2X, age of information, and mobile edge computing. He serves as a reviewer for \textit{ACM Computing Surveys, IEEE Transactions on Mobile Computing, IEEE Journal on Selected Areas in Communications, IEEE Transactions on Intelligent Transportation Systems, IEEE Internet of Things Journal, IEEE Transactions on Vehicular Technology}, and \textit{IEEE Vehicular Technology Magazine}.
\end{IEEEbiography}

\begin{IEEEbiography}[{\includegraphics[width=1in,height=1.25in,clip,keepaspectratio]{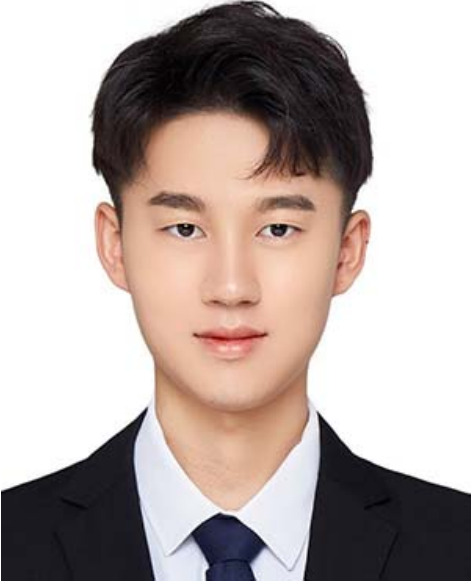}}]{Senkang Hu}
received his B.S. degree in electronic and information engineering from Beijing Institute of Technology, Beijing, China, in 2022. He is currently pursuing his PhD degree in the Department of Computer Science at City University of Hong Kong, Hong Kong. His research interests include autonomous driving, vehicle-to-vehicle collaborative perception.
\end{IEEEbiography}

\begin{IEEEbiography}
[{\includegraphics[width=1in,height=1.25in,clip,keepaspectratio]{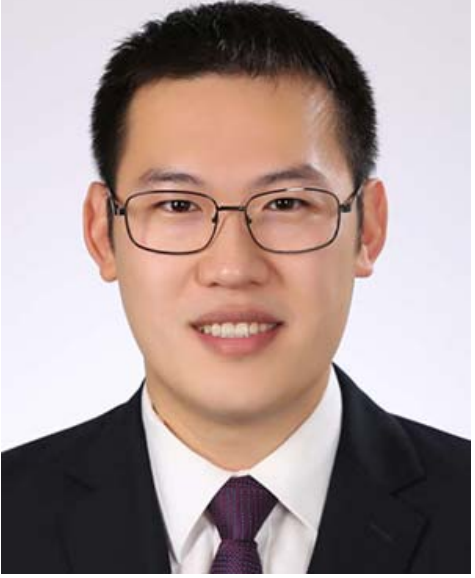}}]{\textbf{Jingjing Wang}} (S'14-M'19-SM'21) received his B.S. degree in Electronic Information Engineering from Dalian University of Technology, Liaoning, China in 2014 and the Ph.D. degree in Information and Communication Engineering from Tsinghua University, Beijing, China in 2019, both with the highest honors. From 2017 to 2018, he visited the Next Generation Wireless Group chaired by Prof. Lajos Hanzo, University of Southampton, UK. Dr. Wang is currently an associate professor at School of Cyber Science and Technology, Beihang University. His research interests include AI enhanced next-generation wireless networks, swarm intelligence and confrontation. He has published over 100 IEEE Journal/Conference papers. Dr. Wang was a recipient of the Best Journal Paper Award of IEEE ComSoc Technical Committee on Green Communications \& Computing in 2018, the Best Paper Award of IEEE ICC and IWCMC in 2019.
\end{IEEEbiography}

\begin{IEEEbiography}
[{\includegraphics[width=1in,height=1.25in,clip,keepaspectratio]{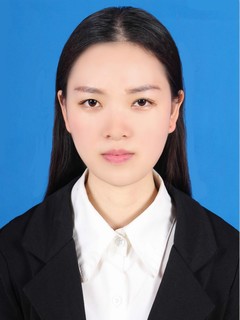}}]{\textbf{Yiqin Deng}} received her MS degree in software engineering and her PhD degree in computer science and technology from Central South University, Changsha, China, in 2017 and 2022, respectively. She is currently a Postdoctoral Researcher with the Department of Computer Science at City University of Hong Kong. Previously, she was a Postdoctoral Research Fellow with the School of Control Science and Engineering, Shandong University, Jinan, China. She was a visiting researcher at the University of Florida, Gainesville, Florida, USA, from 2019 to 2021. Her research interests include edge/fog computing, computing power networks, Internet of Vehicles, and resource management.
\end{IEEEbiography}

\begin{IEEEbiography}[{\includegraphics[width=1in,clip,keepaspectratio]{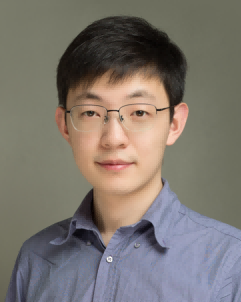}}]{Xianhao Chen}(Member, IEEE) received the B.Eng. degree in electronic information from Southwest Jiaotong University in 2017, and the Ph.D. degree in electrical and computer engineering from the University of Florida in 2022. He is currently an assistant professor at the Department of Electrical and Electronic Engineering, the University of Hong Kong, where he directs the Wireless Information \& Intelligence (WILL) Lab. He serves as a TPC member of several international conferences and an Associate Editor of ACM Computing Surveys. He received the Early Career Award from the Research Grants Council (RGC) of Hong Kong in 2024, the ECE Graduate Excellence Award for research from the University of Florida in 2022, and the ICCC Best Paper Award in 2023. His research interests include wireless networking, edge intelligence, and machine learning.
\end{IEEEbiography}

\begin{IEEEbiography}[{\includegraphics[width=1in,height=1.25in,clip,keepaspectratio]{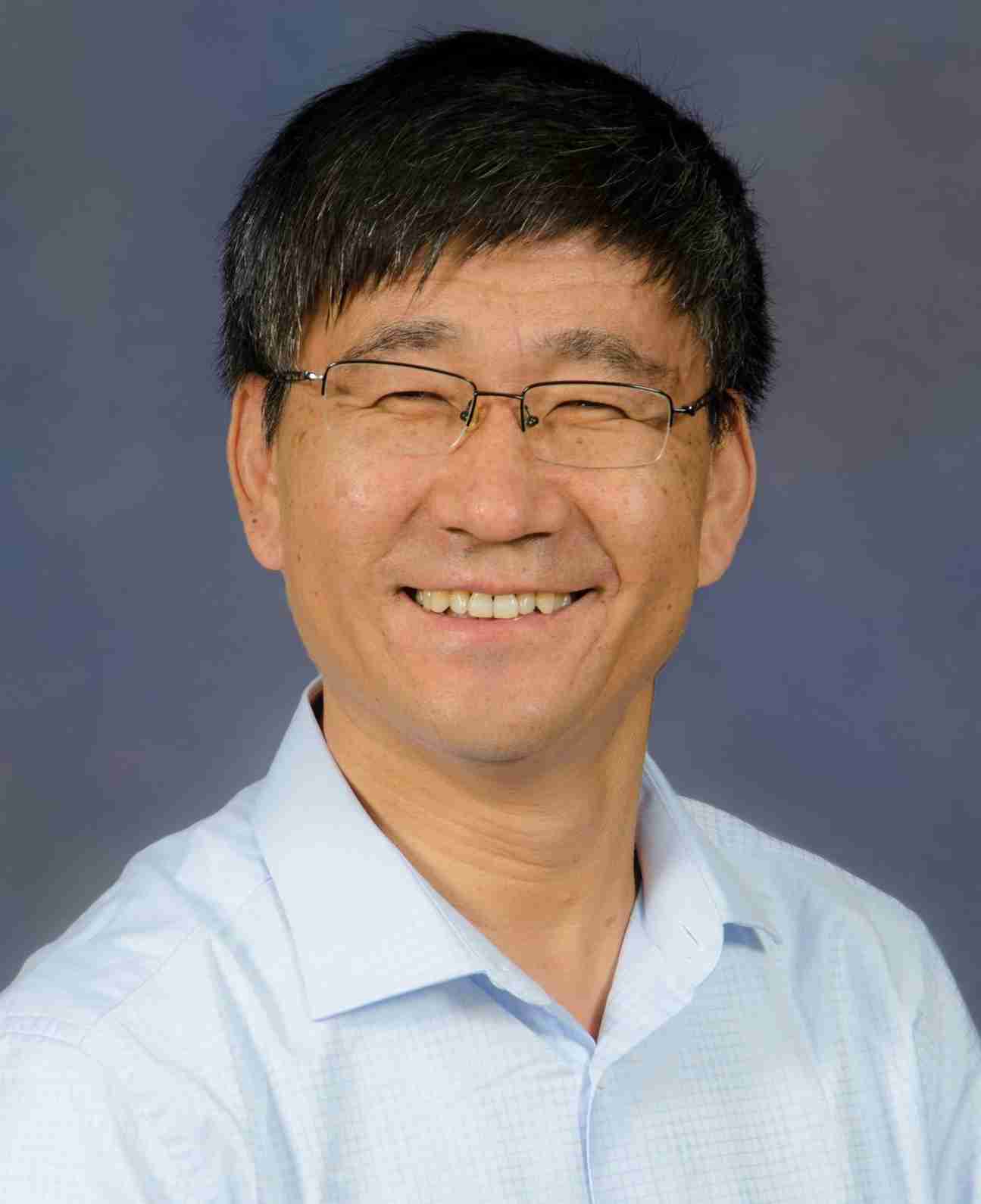}}]{Yuguang Fang}
(S'92, M'97, SM'99, F'08) received
the MS degree from Qufu Normal University, China
in 1987, a PhD degree from Case Western Reserve
University, Cleveland, Ohio, USA, in 1994, and a
PhD degree from Boston University, Boston, Massachusetts, USA in 1997. He joined the Department
of Electrical and Computer Engineering at University of Florida in 2000 as an assistant professor,
then was promoted to associate professor in 2003,
full professor in 2005, and distinguished professor in
2019, respectively. Since August 2022, he has been a Global STEM Scholar and 
Chair Professor with the Department of Computer
Science at City University of Hong Kong.

Prof. Fang received many awards including the US NSF CAREER Award
(2001), US ONR Young Investigator Award (2002), 2018 IEEE Vehicular Technology Outstanding Service Award, IEEE Communications Society
AHSN Technical Achievement Award (2019), CISTC Technical Recognition
Award (2015), and WTC Recognition Award (2014), and 2010-2011 UF
Doctoral Dissertation Advisor/Mentoring Award. He held multiple professorships including the Changjiang Scholar Chair Professorship (2008-2011),
Tsinghua University Guest Chair Professorship (2009-2012), University of
Florida Foundation Preeminence Term Professorship (2019-2022), and University of Florida Research Foundation Professorship (2017-2020, 2006-
2009). He served as the Editor-in-Chief of IEEE Transactions on Vehicular
Technology (2013-2017) and IEEE Wireless Communications (2009-2012)
and serves/served on several editorial boards of journals including Proceedings
of the IEEE (2018-present), ACM Computing Surveys (2017-present), ACM
Transactions on Cyber-Physical Systems (2020-present), IEEE Transactions
on Mobile Computing (2003-2008, 2011-2016, 2019-present), IEEE Transactions on Communications (2000-2011), and IEEE Transactions on Wireless
Communications (2002-2009). He served as the Technical Program Co-Chair of IEEE INFOCOM'2014. He is a Member-at-Large of the Board of
Governors of IEEE Communications Society (2022-2024) and the Director of
Magazines of IEEE Communications Society (2018-2019). He is a fellow of
ACM and AAAS.
\end{IEEEbiography}

{\color{black}{
\begin{appendices}
\section{The Network Architecture of PIB}\label{appendix:architecture}

The PIB framework is designed with efficient computational distribution across the camera and edge server to achieve low latency and high accuracy. As illustrated in Figs. \ref{fig:encoder} and \ref{fig:decoder}, the detailed network architecture is given as follows:

\textbf{Camera:} The camera executes the first two stages of the pipeline: (i) \emph{Feature Extraction} and (ii) \emph{Hyper Encoder}. These stages preprocess the raw video data into a compressed intermediate representation suitable for transmission to the edge server.

\textbf{Edge Server:} Upon receiving the compressed bitstream, the edge server executes (iii) \emph{Hyper Decoder}, (iv) \emph{Projection and Multiview Aggregation}, and (v) \emph{Spatial Aggregation and Classification}. These stages reconstruct the feature maps, fuse multiview information, and generate the pedestrian occupancy map. Below is the detailed breakdown of each stage:

\textbf{(i) Feature Extraction (ResNet-18 Backbone)}:
The feature extraction employs a modified ResNet-18 backbone to retain spatial resolution critical for subsequent projection and fusion. We assume that \(B\) denotes the batch size, \(H\) and \(W\) are the height and width of the input image.
\[
\begin{array}{|c|c|c|}
\hline
\text{Layer Name} & \text{Input Dimensions} & \text{Output Dimensions} \\
\hline
\text{Input Image} & [B, 3, H, W] & [B, 3, 720, 1280] \\
\text{ResNet-18 (Part 1)} & [B, 3, 720, 1280] & [B, 64, 180, 320] \\
\text{ResNet-18 (Part 2)} & [B, 64, 180, 320] & [B, 512, 90, 160] \\
\text{Feature Extraction} & [B, 512, 90, 160] & [B, 8, 90, 160] \\
\hline
\end{array}
\]
\textbf{(ii) Hyper Encoder and  (iii) Decoder for Compression}:
The Hyper Encoder compresses the extracted features at the camera, while the Hyper Decoder reconstructs them at the edge server.
\[
\begin{array}{|c|c|c|}
\hline
\text{Layer Name} & \text{Input Dimensions} & \text{Output Dimensions} \\
\hline
\text{Hyper Encoder} & [B, 8, 90, 160] & [B, 4, 30, 40] \\
\text{Hyper Decoder} & [B, 4, 30, 40] & [B, 8, 90, 160] \\
\hline
\end{array}
\]
\textbf{(iv) Projection and Multiview Aggregation}:
Feature maps are projected onto a common ground plane and aggregated with coordinate maps for multiview fusion. \(H_g\) and \(W_g\) are the height and width of the projected ground plane grid.
\[
\begin{array}{|c|c|c|}
\hline
\text{Layer Name} & \text{Input Dimensions} & \text{Output Dimensions} \\
\hline
\text{Projection} & [B, 8, 90, 160] & [B, 8, H_g, W_g] \\
\text{Concatenation} & [B, 8, H_g, W_g] & [B, N \times 8 + 2, H_g, W_g] \\
\hline
\end{array}
\]
\textbf{(v) Spatial Aggregation and Classification}:
The aggregated features are processed to produce the final pedestrian occupancy map.
\[
\begin{array}{|c|c|c|}
\hline
\text{Layer Name} & \text{Input Dimensions} & \text{Output Dimensions} \\
\hline
\text{Map Classifier} & [B, N \times 8 + 2, H_g, W_g] & [B, 1, H_g, W_g] \\
\hline
\end{array}
\]
\end{appendices}
}}

\end{document}